\documentclass[reqno]{amsart}
\usepackage[foot]{amsaddr}
\usepackage{graphicx}
\usepackage[margin=3cm]{geometry}
\usepackage{amsmath, amssymb,amsthm}
\usepackage[scr=rsfs]{mathalpha}
\usepackage[shortlabels]{enumitem}
\usepackage{mlmodern}
\DeclareSymbolFont{largesymbols}{OMX}{cmex}{m}{n} % prevent extra large \sqrt from mlmodern
\usepackage{mathtools} % multlined
\usepackage{tikz}
\usetikzlibrary{decorations.pathreplacing,decorations.pathmorphing,arrows.meta}
\usepackage{upref} % \ref is upright even in italicized environments
\usepackage[colorlinks]{hyperref}
\hypersetup{citecolor=blue,filecolor=blue,linkcolor=blue,urlcolor=navyblue}
\definecolor{navyblue}{rgb}{0.0, 0.0, 0.5}

\newtheorem{thm}{Theorem}[section]
\newtheorem{prop}[thm]{Proposition}
\newtheorem{lem}[thm]{Lemma}

\newtheorem{conj}{Conjecture}
\newtheorem{assumption}{Assumption}
\newtheorem*{conj*}{Conjecture}
\newtheorem{problem}{Problem}
\theoremstyle{definition}

\newtheorem{rmk}{Remark}[section]

\newtheorem*{claim*}{Claim}

\newcommand{\N}{\mathbb{N}}
\newcommand{\Z}{\mathbb{Z}}
\newcommand{\R}{\mathbb{R}}
\newcommand{\C}{\mathbb{C}}

\newcommand{\Sc}{\mathbb{S}_\C}

\newcommand{\E}{\mathbf{E}}
\newcommand{\oneb}{\mathbf{1}}
\renewcommand{\P}{\mathbf{P}}

\newcommand{\CN}{\mathcal{CN}}
\newcommand{\RN}{\mathcal{N}}

\newcommand{\Tr}{\operatorname{Tr}}

\newcommand{\dsim}{\sim}
\newcommand{\rvsim}{\overset{d}{=}}

\newcommand{\tv}{\mathrm{TV}}
\newcommand{\G}{\mathcal{G}}

\newcommand{\ggtnk}{\G\G^T_{NK}}
\newcommand{\gsym}{\G_N^\mathrm{sym}}

\newcommand{\Gb}{\mathbf{G}}
\newcommand{\dkl}{D_{\mathrm{KL}}}
\newcommand{\rank}{\operatorname{rank}}
\newcommand{\qf}{q}
\newcommand{\poly}{\operatorname{poly}}
\newcommand{\Haf}{\operatorname{Haf}}
\newcommand{\Per}{\operatorname{Per}}
\newcommand{\hs}{\mathrm{hs}}
\newcommand{\Beta}{\operatorname{Beta}}
\newcommand{\Unif}{\operatorname{Unif}}

\newcommand{\sans}[1]{\text{\normalfont\fontfamily{lmss}\selectfont #1}}
\newcommand{\sharpp}{\sans{\#P}}
\newcommand{\bpp}{\sans{BPP}}
\newcommand{\fbpp}{\sans{FBPP}}
\newcommand{\postbpp}{\sans{PostBPP}}
\newcommand{\fpostbpp}{\sans{FPostBPP}}
\newcommand{\np}{\sans{NP}}
\newcommand{\p}{\sans{P}}

\newcommand{\sym}{\mathrm{Sym}}
\newcommand{\event}{\mathcal{E}}

\newcommand{\U}{\mathrm{U}}
\newcommand{\Uf}{\mathrm{U}'}

\newcommand\numberthis{\stepcounter{equation}\tag{\theequation}}
\numberwithin{equation}{section}

\begin{document}
\title[]{Proof of the hiding conjecture for Gaussian boson sampling with an arbitrary number of squeezed input modes}

\author{
Laura Shou$^{1}$,
Alexey V. Gorshkov$^{1,3}$,
Victor Galitski$^1$,
Sarah H. Miller$^{2}$
}

\address{\normalfont$^1$Joint Quantum Institute, Department of Physics, NIST/University of Maryland, College Park, MD 20742, USA}

\address{\normalfont$^2$Applied Research Laboratory for Intelligence and Security,
University of Maryland, College Park, Maryland 20742, USA}

\address{\normalfont$^3$Joint Center for Quantum Information and Computer Science, NIST/University of Maryland, College Park, MD, 20742, USA}

\begin{abstract}
Gaussian boson sampling (GBS) is a sampling task proposed to demonstrate quantum advantage.
We consider Gaussian boson sampling on $M$ optical modes, with $K$ equally squeezed input modes and $N$ observed photon counts.
We complete the proof of the hiding conjecture for Gaussian boson sampling with an arbitrary number of squeezers $K$, which is a part of the argument for classical hardness of GBS. In particular, we show that for any $K$ and $N=o(\sqrt{K})$, the symmetric product $MK^{-1/2}U_{NK}U_{NK}^T$, for $U_{NK}$ the top left $N\times K$ submatrix of an $M\times M$ Haar random unitary $U$, is close in total variation distance to both an $N\times N$ symmetric complex Gaussian matrix $\Gb$ with independent entries, and the symmetric product $GG^T/\sqrt{K}$ for $G$ an $N\times K$ matrix of iid standard complex Gaussians.
We show however that the density-based instance generating method of \cite[Lemma 5.8]{aa} used to efficiently implement a hiding procedure fails for Gaussian boson sampling with $K=cM$ if $c<1/2$. Instead we use approximate instance generating to implement the hiding for the usual classical hardness reduction.
\end{abstract}

\maketitle

\section{Introduction}

Gaussian boson sampling \cite{HamiltonGBS2017} is a sampling task which is expected to be hard for classical computers, but currently realizable in existing quantum experiments.
In this task, one prepares an initial Gaussian state consisting of $M$ single-mode squeezed vacuum states with squeezing parameters $s_i\ge0$. For simplicity we take the first $K$ modes to have the same squeezing parameters $s_i=s>0$, and the remaining $M-K$ modes to have the vacuum state.
The initial state is then inserted into an $M$-mode passive linear optical network described by an $M\times M$ linear optical unitary $U$ (Figure~\ref{fig:gbs}). After interfering in the linear optical network, the resulting output state is measured in the photon-number basis, producing a photon count outcome $\mathbf n\in\{0,1,2,\ldots\}^M$, with total photon number $N=\sum_{i=1}^M\mathbf n_i$.
Let $I_K$ be the $M\times M$ diagonal matrix whose first $K$ diagonal entries are 1s, while the rest are 0s. 
Consider the $N\times N$ submatrix $(UI_KU^T)_{\mathbf n,\mathbf n}$ of the $M\times M$ matrix $UI_KU^T$ formed by taking the rows and columns corresponding to $1$s in $\mathbf n$.
For collision-free outcomes $\mathbf n\in\{0,1\}^M$, the probability of observing $\mathbf n$ is \cite{HamiltonGBS2017,KruseGBS2019}
\begin{align}\label{eqn:hprob}
\P[\mathbf n]&=\frac{\tanh^N(s)}{\cosh^K(s)}|\Haf[(UI_KU^T)_{\mathbf n,\mathbf n}]|^2,
\end{align}
where $\Haf[X]$ denotes the \emph{hafnian} of a symmetric matrix $X$; letting $N=2n$,
the hafnian is $\Haf[X]=\sum_{\pi\in\mathcal P_2(2n)}\prod_{\{i,j\}\in\pi}X_{ij}$, where $\mathcal P_2(2n)$ is the set of all pairings (i.e.~perfect matchings) of $2n$ elements. 
Conditioned on observing total photon number $N$, the collision-free condition for Haar random $U$ occurs with high probability if $N=o(M^{1/2})$ \cite{aa,deshpande2022quantum}. Since the average number of photons is $\langle N\rangle=K\sinh^2s$, one takes the squeezing parameter $s$ to be small to ensure $\langle N\rangle=o(M^{1/2})$.

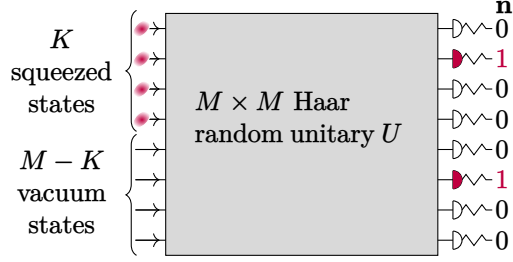
\begin{figure}[htb]
\begin{tikzpicture}[scale=.8]
\foreach \mode in {1,...,8}{
\draw (.5,-\mode/2) --++ (3.5,0);
\draw[->] (.5,-\mode/2)--++(.4,0);
% right side
\draw[xshift=4cm] (.5,-\mode/2) --++ (1.25,0);
\draw[xshift=4cm,yshift=-.15cm] (1.75,-\mode/2) arc (-90:90:.15);
\draw[xshift=4cm,yshift=-.15cm] (1.75,-\mode/2)--++(0,.3);
\draw[xshift=4cm,  decorate, decoration={zigzag, segment length=6pt, amplitude=2pt}] (1.9,-\mode/2)--++(.5,0);
}
\foreach \mode in {1,...,4}{
\draw[inner color=purple, rotate around={45:(.6,-\mode/2)},color=white] (.6,-\mode/2) ellipse (1.75mm and 1mm);
}
\foreach \mode in {2,6}{
\fill[purple,xshift=4cm,yshift=-.15cm] (1.75,-\mode/2) arc (-90:90:.15);
}
\draw[fill=gray!30] (1,-4.25)--(1,-.25)--(5.5,-.25)--(5.5,-4.25)--cycle;
\node at (3.2,-2) {\parbox{2.75cm}{$M\times M$ Haar\\ random unitary $U$}};
\foreach \mode in {1,3,4,5,7,8}{
\node[right] at (6.3,-\mode/2) {$0$};
}
\foreach \mode in {2,6}{
\node[right, purple] at (6.3,-\mode/2) {$1$};
}
\node[above right] at (6.3,-.4) {$\mathbf{n}$};
%labels
\draw[decoration={brace,raise=3pt,aspect=.5,amplitude=4pt},decorate,xshift=1.5mm] (0.5,-2.2)--++(0,2);
\node[left] at (1,-1.2)  {\parbox{2.5cm}{\centering $K$ \\squeezed\\states}};

\draw[decoration={brace,raise=3pt,aspect=.5,amplitude=4pt,mirror},decorate,xshift=1.5mm] (0.5,-2.25)--++(0,-2);
\node[left] at (1,-3.2) {\parbox{2.5cm}{\centering $M-K$ \\vacuum\\states}};
\end{tikzpicture}
\caption{Illustration of Gaussian boson sampling experiment. $K$ single-mode squeezed states are inserted into an $M$-mode linear optical network described by an $M\times M$ linear optical unitary $U$. The output state is measured in the photon number basis, producing a random sample $\mathbf n\in\{0,1,2,\ldots\}^M$ of photon counts. For collision-free outcomes, $\mathbf n\in\{0,1\}^M$.}\label{fig:gbs}
\end{figure}

The argument that it is classically hard to generate samples $\mathbf n$ according to the distribution \eqref{eqn:hprob} is based on classical hardness of exactly computing permanents and hafnians of complex matrices \cite{Valiant1979}.
However, since a realistic Gaussian boson sampler will always have some amount of noise and errors, one has to consider the task of \emph{approximate sampling} from the distribution described by \eqref{eqn:hprob}.
The hardness of approximate average-case sampling then relies on two properties:
\begin{enumerate}
\item A complexity theoretic conjecture that \emph{approximating} the hafnian of a random complex Gaussian-type matrix to a certain additive error is \sharpp-hard in the average-case. 
\item A hiding conjecture, which essentially states that one can ``hide'' a random complex Gaussian matrix $X$ as a submatrix of $UI_KU^T$ in total variation distance 
(Conjecture~\ref{conj:hiding}). This would allow one to use a GBS oracle running $U$ to approximate $|\Haf(X)|^2$ in $\bpp^\np$.
\end{enumerate}

We focus on the hiding conjecture (2).
Previously, only special cases of $K$ were proved to satisfy the hiding property, namely in the sparse squeezer regime $NK=o(M)$ \cite{aa,deshpande2022quantum,shou2026proof}, and the case $K=M$ \cite{shou2026proof}. Here we prove the hiding conjecture for all other $K$, which proves the full range for the hiding conjecture for Gaussian boson sampling. 

\subsection{Preliminaries}

To state the precise hiding conjecture and results, we first recall some definitions.
The total variation distance (TVD) between two probability  measures $\mu$ and $\nu$ on a measure space $(E,\mathcal E)$ is
\begin{align}\label{eqn:tv}
d_\tv(\mu,\nu)&\equiv\sup_{A\in\mathcal E}|\mu(A)-\nu(A)|.
\end{align}
It also has the characterization
\begin{align}\label{eqn:tvd-f}
d_\tv(\mu,\nu)&=\sup_{\|f\|_\infty\le1}\frac12\int f\,(d\mu-d\nu).
\end{align}
If $\mu$ and $\nu$ have densities $p$ and $q$ with respect to a measure $d\lambda$ on $E$, then also
\begin{align}\label{eqn:tv2}
d_\tv(\mu,\nu)&=\frac{1}{2}\int_E|p(x)-q(x)|\,d\lambda(x).
\end{align}
For random variables $Z$ and $W$, we write $d_\tv(Z,W)$ to mean the total variation distance between their distributions $\mathcal L(Z)$ and $\mathcal L(W)$.

For any coupling $(X,Y)$ of $(\mu,\nu)$, i.e. random variables $X,Y$ such that $X\dsim\mu$ and $Y\dsim\nu$, there is the TVD coupling bound $d_\tv(\mu,\nu)\le\P[X\ne Y]$.

We will use the following probability distributions.
\begin{itemize}
\item Let $\CN(0,\sigma^2)$ denote the complex Gaussian distribution whose real and imaginary parts are independent Gaussians with mean $0$ and variance $\sigma^2/2$.
\item Let $\gsym$ denote the ensemble of $N\times N$ symmetric random matrices $\Gb$ with $\CN(0,2)$ diagonal entries and $\CN(0,1)$ off-diagonal entries, with all entries independent modulo the symmetry requirement. We will use bold font $\Gb$ only for $\Gb\dsim\gsym$.
\item Let $\ggtnk$ be the ensemble of $N\times N$ matrices $GG^T/\sqrt{K}$ where $G$ is an $N\times K$ matrix of iid standard complex Gaussian entries. The normalization is chosen so that $GG^T/\sqrt{K}$ has a nondegenerate limiting distribution as $K\to\infty$. Note that since $G$ is complex, $GG^T$ is \emph{not} Wishart.
\item Let $\mathcal H_M$ denote the Haar measure on the unitary group $\U(M)$.
\end{itemize}

The hiding conjecture can be stated as follows \cite{aa,deshpande2022quantum,shou2026proof}:
\begin{conj}[hiding in Gaussian boson sampling]\label{conj:hiding}
Let $U_{NK}$ be the top left $N\times K$ submatrix of an $M\times M$ Haar random unitary matrix, and let $Z=Z_{N,K}$ be a matrix with distribution given by either $\gsym$ or $\ggtnk$.
Then for $N\le K\le M$, there exist polynomials $p,r$ such that for any $\delta>0$ and $M\ge p(N)/r(\delta)$, at least one of the choices of $Z$ satisfies
\begin{align}
d_\tv(MK^{-1/2}U_{NK}U_{NK}^T,Z)=O(\delta), \label{eqn:dtv}
\end{align}
where $d_\tv(\cdot,\cdot)$ denotes the total variation distance as defined in \eqref{eqn:tv}.
\end{conj}

The sparse squeezer case $N=K=o(M^{1/5})$ was observed \cite{HamiltonGBS2017,deshpande2022quantum} to follow from hiding in Fock boson sampling \cite{aa} with $Z\dsim\ggtnk$. However, the non-sparse regime where $K$ is large is of most interest experimentally \cite{zhong2020quantum, Zhong2021, madsen2022, deng2023gaussian, liu2026gaussian}. Additionally, a large number of squeezers $K$ is favorable for the anticoncentration results of \cite{Ehrenberg2025transition,Ehrenberg2025second}, which provide evidence for sampling hardness. 
In the large $K\propto M$ regime, only the case $K=M$ was proved, with $Z\dsim\gsym$, using that in this case the matrix $U_{NK}U_{NK}^T$ is a submatrix of a COE (Circular Orthogonal Ensemble) random matrix  \cite{shou2026proof}. 

\subsection{Main results}

In this paper, we prove the hiding conjecture for Gaussian boson sampling for any number of input squeezed modes $K$, where $N\le K\le M$. This fully resolves Conjecture~\ref{conj:hiding}. Additionally, we prove TVD closeness of the distributions $\gsym$ and $\ggtnk$ in the regime $N=o(\sqrt{K})$, so that Conjecture~\ref{conj:hiding} holds with either distribution in this regime.

Recall in the context of Gaussian boson sampling, $M$ is the number of modes, $K$ is the number of squeezed modes, and $N$ is the number of detected photons. While $N$ is even for Gaussian boson sampling, we do not require $N$ to be even in Theorems~\ref{thm:hiding} and \ref{thm:ggt} below.

\begin{thm}[hiding in Gaussian boson sampling]\label{thm:hiding}
Let $N\le K\le M$, let $U_{NK}$ be the top left $N\times K$ submatrix of an $M\times M$ Haar random unitary matrix $U$, and let $\Gb\dsim\gsym$. Then as $K\to\infty$,
\begin{align}\label{eqn:tvhiding}
d_\tv(MK^{-1/2}U_{NK}U_{NK}^T,\Gb)=O\left(\frac{N}{\sqrt{K}}\right),
\end{align}
which is $o(1)$ if $N=o(\sqrt{K})$.
\end{thm}

\begin{rmk}
\begin{enumerate}[(i)]

\item Theorem~\ref{thm:hiding} combined with the sparse case result \cite[Theorem 1.5]{shou2026proof} for $NK=o(M)$ gives the full range of Conjecture~\ref{conj:hiding} by choosing polynomials $p,r$ appropriately. 
For example:
\begin{itemize}
\item For $K\propto M$, one can take $Z\dsim\gsym$, $p(N)=N^2$, and $r(\delta)=c\delta^2$ for sufficiently small $c$, by Theorem~\ref{thm:hiding}. 
\item For $K\le M^{1/2}$, one can take $Z\dsim\ggtnk$, $p(N)=N^{2+\epsilon}$, and $r(\delta)=c\delta^4$ by the sparse result, rewritten as \eqref{eqn:sparse}.
\item For any $K\ge M^{1/2}$, one can take $Z\dsim\gsym$, $p(N)=N^4$, and $r(\delta)=c\delta^4$, again by Theorem~\ref{thm:hiding}.
\end{itemize}
In fact, using Theorem~\ref{thm:ggt} below, we also obtain Conjecture~\ref{conj:hiding} with a fixed distribution $Z=GG^T/\sqrt{K}\dsim\ggtnk$ and fixed polynomials $p,r$, over any $N\le K\le M$: For $p(N)=N^4$, $r(\delta)=c\delta^4$, and any $M\ge p(N)/r(\delta)$,
\begin{align}
d_\tv(MK^{-1/2}U_{NK}U_{NK}^T,GG^T/\sqrt{K})=O(\delta).
\end{align}

\item In the main regime of interest $K\propto M$, the condition $N=o(\sqrt{K})$ becomes $N=o(\sqrt{M})$, which is the conjectured maximal submatrix size allowed \cite{deshpande2022quantum}. As we will show in Theorem~\ref{thm:ggt}, in the regime $N=o(\sqrt{K})$, TVD closeness in \eqref{eqn:dtv} actually holds with either distribution $Z\dsim\gsym$ or $\ggtnk$.

\item The proof of Theorem~\ref{thm:hiding} relies on relating the $K<M$ case to the $K=M$ case. One could alternatively write down the density formula for $U_{NK}U_{NK}^T$ in terms of a matrix integral and use quantitative Laplace approximation to bound the TVD. However, due to error bounds for high-dimensional Laplace's method, we do not expect this approach, at least with standard estimates, to obtain any sharper results.

\end{enumerate}
\end{rmk}

Let $G$ be an $N\times K$ matrix of iid standard complex Gaussians.
Combining Theorem~\ref{thm:hiding} with TVD closeness of $MU_{NK}U_{NK}^T$ to $GG^T$ in a ``sparse squeezer'' regime with $N\le K$ and $NK=o(M)$ \cite{shou2026proof}, we will obtain
\begin{thm}[$GG^T$ vs $\gsym$]\label{thm:ggt}
Let $G$ be an $N\times K$ matrix of iid standard complex Gaussians, and let $\Gb\dsim\gsym$. Then as $K\to\infty$,
\begin{align}
d_\tv(GG^T/\sqrt{K},\Gb)=O\left(\frac{N}{\sqrt{K}}\right),
\end{align}
which is $o(1)$ if $N=o(\sqrt{K})$.
\end{thm}
This implies that when $N=o(\sqrt{K})$, it does not matter whether we use $GG^T$ or $\Gb\dsim\gsym$ as the target Gaussian-type distribution in the hiding statement. 
We note that the matrices $\Gb$ are much nicer to work with due to their independent entries. Moreover, they are much closer to the random Gaussian matrices used in the hardness reduction for Fock boson sampling, which suggests techniques and results for hardness of Fock boson sampling will carry over more easily to $\Gb$ than to $GG^T$.

\begin{rmk}
The $N=o(\sqrt{K})$ condition differs from the $\Theta(K^{1/3})$ transition for GOE (Gaussian Orthogonal Ensemble) behavior of Wishart matrices with $K$ degrees of freedom \cite{bubeck2016testing,jiang2015approximation,RaczRichey2019}. It remains to locate the precise location of the transition in this case.
\end{rmk}

The TVD hiding property Theorem~\ref{thm:hiding} demonstrates it is possible to ``hide'' a Gaussian matrix $\Gb\dsim\gsym$ as a random instance of $MK^{-1/2}(UI_KU^T)_{S,S}$ for $U$ a Haar random unitary matrix and $S$ a random size $N$ subset of $\{1,\ldots,M\}$, up to small TVD error. 
However, for the usual classical hardness argument, given $\Gb$, we need to generate an instance of such a random $U$ and $S^*$ efficiently.
The instance-generating procedure for Fock boson sampling uses a rejection sampling method based on the density bound $\tilde f(Z)\le (1+o(1))\tilde g(Z)$, where $\tilde f$ is the unitary submatrix density and $\tilde g$ the Gaussian density \cite[Lemma 5.7]{aa}. However, we show that such a pointwise density bound cannot hold for Gaussian boson sampling with $K=\alpha M$, $0<\alpha<1/2$.

\begin{prop}[no density ratio bound]\label{prop:inf}
Let $K\ge4$, $K\ge N$, and $N+K\le M$. Denote by $U_{NK}$ the top left $N\times K$ submatrix of an $M\times M$ Haar random unitary matrix $U$, and let $\Gb\dsim\gsym$.
Let $f$ be the density function of $MK^{-1/2}U_{NK}U_{NK}^T$ (if it exists) and $g$ be the density function of $\Gb$ over the space of $N\times N$ complex symmetric matrices\footnote{We view this as the space spanned by the upper triangular matrix elements}. Then for $K=\alpha M$, $0<\alpha<1/2$, and any $N\le K$, either the density function doesn't exist, or
\begin{align}\label{eqn:densityratio}
\operatorname{ess\,sup}_{Z}\frac{f(Z)}{g(Z)}\ge\Omega_\alpha(e^{c_\alpha M}),
\end{align}
for a constant $c_\alpha>0$ and where $\Omega_\alpha$ indicates the implicit constant may also depend on $\alpha$. 

In particular, consider a sequence of $K=K(M),N=N(M)$ with $\alpha=\alpha(M)=K(M)/M$ bounded in $[\eta,1/2-\eta]$ for some $\eta>0$, and let $f_{M,K,N},g_{M,K,N}$ denote the density functions as above if they exist. Then 
\begin{align}
\liminf_{M\to\infty}\operatorname{ess\,sup}_{Z}\frac{f_{M,K,N}(Z)}{g_{M,K,N}(Z)}=+\infty.
\end{align}
\end{prop}

\begin{rmk}
This result is perhaps counterintuitive, since intuitively, smaller $K\propto M$ should not make the submatrix $U_{NK}U_{NK}^T$ look \emph{less} Gaussian than the $K=M$ case, which sees the full orthonormality requirements of $U$. In the $K=M$ case, the bound $f(Z)\le (1+o(1))g(Z)$ holds \cite{shou2026proof} for $N=o(K^{1/3})$. However, the pointwise density bound \eqref{eqn:densityratio} captures rare tail behavior, which is not necessarily captured by TVD, and which plausibly can differ for smaller $K$ where $MK^{-1/2}U_{NK}U_{NK}^T$ may have different tail behavior.
\end{rmk}

Proposition~\ref{prop:inf} means we cannot use the density-based rejection sampling method of \cite[Lemmas 5.7, 5.8]{aa}, which requires $f(Z)\le (1+\delta)g(Z)$, to generate instances of unitaries $U$ with $\Gb$ hidden as an instance of $MK^{-1/2}(UI_KU^T)_{S^*,S^*}$ when $K<M/2$. One could perhaps try to prove the required density bound for $K>M/2$ (we know it at least holds for $K=M$ and $N=o(K^{1/3})$), or try to truncate the distribution $\Gb$ to remove tail behavior, but both of these approaches are model-specific and also likely involve working with complicated density functions.
Instead, we adapt the argument of \cite[\S5.2]{aa} to generate \emph{approximately} Haar instances of unitaries $U$ and approximately uniform locations $S$ to hide $\Gb$, using postselection with an \np\ oracle. This will be enough to complete the hardness reduction, under a finite-precision implementation assumption, giving Theorem~\ref{thm:hardness} below.
To state the theorem, first define
\begin{problem}[$|\mathrm{GHE}|_\pm^2$]\label{prob:ghe}
Let $N\in2\N$.
Given as input a matrix $X\dsim\gsym$, together with error bounds $\varepsilon,\delta>0$, estimate $|\Haf(X)|^2$ to within additive error $\pm\varepsilon\cdot\E|\Haf(X)|^2$ with probability at least $1-\delta$ over $X$ and the algorithm's randomness in $\poly(N,1/\varepsilon,1/\delta)$ time.
\end{problem}

For $X\dsim\gsym$, one can use independence of entries above the diagonal to quickly calculate that
\begin{align}
\E|\Haf(X)|^2&=(N-1)!!=\frac{N!}{(N/2)!2^{N/2}}\sim\frac{\sqrt{2}N^{N/2}}{e^{N/2}}.
\end{align}
As is usual \cite[\S2]{aa}, it will be understood that all entries of $X$ are rounded to polynomially many bits of precision; here we allow $\poly(M,1/\delta,1/\varepsilon)$ bits of precision, for a sufficiently large fixed polynomial.
To formalize this in our case, we will make an assumption on finite-precision implementation, stated later precisely as Assumption~\ref{finite} in Section~\ref{sec:hardness}. The assumption essentially says that we can approximate Haar random $U$ by finite-precision descriptions $v=v_\xi$, and also $(UI_KU^T)_{S,S}$ by an efficiently computable finite-precision matrix $\hat Z$, using high enough numerical precision. We expect this assumption is true and that it can be proved by careful finite-precision accounting and Gram-Schmidt or QR factorization.

Under this finite-precision implementation assumption, we prove the analogue of the main boson sampling result of \cite[Theorem 1.3]{aa}, for Gaussian boson sampling with essentially arbitrary number of squeezed input modes $K$.
We refer to \cite{complexityzoo} for definitions of the standard complexity classes \np, \bpp\ (bounded-error probabilistic polynomial-time), and \fbpp\ (the function/search analogue of \bpp, which searches for a witness to a relation in probabilistic polynomial time).
The notation $\fbpp^\np$ means \fbpp\ with access to an \np\ oracle.
\begin{thm}[main hardness result]\label{thm:hardness}
Let the probability distribution $\mathcal D_A$ be the output of a Gaussian boson sampling experiment $A=A(v,K,s)$, for $v$ a finite-precision description of a linear optical unitary $U=U(v)$, and $K,s$ input squeezing parameters.
Suppose there exists a classical algorithm $C$ which is able to approximately sample from limited instances of $A$, say for each $M\in\N$, only those with $K=K(M)$ equally squeezed input modes for some given\footnote{We assume the sequence $K(M)$ is efficiently computable.} sequence $K(M)$ with $M=O(\poly(K(M)))$, and all small squeezing parameters $s=O(K^{-1/4})$.
More precisely, the algorithm $C$ takes as input a description of such $A$ as well as an error bound $\varepsilon$, and samples from a probability distribution $\mathcal D'_A$ such that $\|\mathcal D'_A-\mathcal D_A\|_\tv\le\varepsilon$ in $\poly(|A|,1/\varepsilon)$ time, where $|A|$ denotes the length of the description\footnote{As in \cite{aa}, and as discussed more in Section~\ref{sec:hardness}, it will be understood that all entries in $A$ are rounded to $\poly(M,1/\delta,1/\varepsilon)$ bits of precision.} of $A$.
Then under the finite-precision implementation Assumption~\ref{finite}, the problem $|\mathrm{GHE}|_\pm^2$ is solvable in $\fbpp^\np$.
In other words, if we treat $C$ as a black box, then $|\mathrm{GHE}|_\pm^2\in\fbpp^{\np^C}$.
\end{thm}

\begin{conj}\label{conj:ghe}
$|\mathrm{GHE}|_\pm^2$ is \sharpp-hard, in the sense that if $\mathcal O$ is any oracle that solves $|\mathrm{GHE}|_\pm^2$, then $\sans{P}^{\#\sans{P}}\subseteq\bpp^{\mathcal O}$.
\end{conj}
Conjecture~\ref{conj:ghe} is analogous to the conjecture for additive approximation of squared permanents of complex matrices, $|\mathrm{GPE}|_\pm^2$, being \sharpp-hard \cite{aa}. Note that in the formulation of $|\mathrm{GHE}|_\pm^2$ here as well as in $|\mathrm{GPE}|_\pm^2$, the additive error size is $\varepsilon$ times the average value of the quantity to estimate ($|\Haf(X)|^2$ here, $|\Per(X)|^2$ for $|\mathrm{GPE}|_\pm^2$). This makes Problem~\ref{prob:ghe} the natural hafnian analogue of $|\mathrm{GPE}|_\pm^2$ from \cite{aa} (previous hafnian hardness problems in the context of GBS used the more complicated matrix $XX^T$ for $X$ an $N\times K$ matrix of iid complex Gaussians, and also did not express the error bound in terms of the hafnian moments).
If Conjecture~\ref{conj:ghe} holds, then the existence of such a classical algorithm $C$ in Theorem~\ref{thm:hardness} implies collapse of the polynomial hierarchy by Toda's theorem \cite{toda1991pp}.

\subsection{Outline}
The rest of the paper is organized as follows.
Additionally, we note that generic constants $C,c$ may change from line to line throughout the paper.
\begin{itemize}
\item In Section~\ref{sec:hiding}, we prove a simpler version of the hiding property Theorem~\ref{thm:hiding}, though which holds only for $N=o(K^{1/3})$ in general. This proof is less technical, and is enough to perform the subsequent hardness reduction. The extension to $N=o(\sqrt{K})$ is proved in Appendix~\ref{sec:hiding2}.
\item In Section~\ref{sec:ggt}, we prove Theorem~\ref{thm:ggt} on $\gsym$ vs $\ggtnk$.
\item In Section~\ref{sec:hardness}, we prove the hardness reduction Theorem~\ref{thm:hardness}.
\item In Section~\ref{sec:inf}, we prove Proposition~\ref{prop:inf} on the lack of a density ratio bound for $K=\alpha M$, $0<\alpha<1/2$.
\item In Appendix~\ref{sec:hiding2}, we prove the full Theorem~\ref{thm:hiding} up to submatrix size $N=o(\sqrt{K})$.
\end{itemize}

\section{Proof of \texorpdfstring{$N=o(K^{1/3})$}{N=o(K^(1/3))} TVD hiding}\label{sec:hiding}

In this section we prove a weaker version of Theorem~\ref{thm:hiding}; namely, 
\begin{prop}[hiding for $N=o(K^{1/3})$]\label{prop:hiding-3}
Let $N\le K\le M$, let $U_{NK}$ be the top left $N\times K$ submatrix of an $M\times M$ Haar random unitary matrix $U$, and let $\Gb\dsim\gsym$. Then as $K\to\infty$,
\begin{align}\label{eqn:tvhiding-3}
d_\tv(MK^{-1/2}U_{NK}U_{NK}^T,\Gb)&=O\left(\sqrt{\frac{N^3}{K}}\right),
\end{align}
which is $o(1)$ if $N=o(K^{1/3})$.
\end{prop}
This is strictly weaker than the bound and $N=o(K^{1/2})$ allowed in Theorem~\ref{thm:hiding}, but the proof of Proposition~\ref{prop:hiding-3} will be simpler, and this statement is enough to go through the usual hardness reduction with some minor parameter adjustments. (Recall, the original boson sampling hardness reduction of \cite{aa} had $N=o(M^{1/5})$; proving TVD closeness with any inverse polynomial power is sufficient.) The key idea of the proof of Proposition~\ref{prop:hiding-3} or Theorem~\ref{thm:hiding} is Lemma~\ref{lem:reduce}, which is presented in this section.
The proof of the stronger bound in Theorem~\ref{thm:hiding} differs only in a later technical estimate, which we give in Appendix~\ref{sec:hiding2}.

\vspace{2mm}
It suffices to prove \eqref{eqn:tvhiding-3} for $N^3=O(K)$; otherwise the bound is trivial.
Let $\mathcal V$ be the span of the last $M-K$ columns of $U$. The main idea is the following: The first $K$ columns of $U$ form an orthonormal basis for $\mathcal V^\perp$, and conditioned on $\mathcal V$ they form a Haar random basis for $\mathcal V^\perp$ \cite[\S1.2]{Meckes-book}. Since $\mathcal V^\perp$ has dimension $K$, this can be used to relate the distribution of $U_{NK}U_{NK}^T$ to the distribution of a matrix transformation of $W_{NK}W_{NK}^T$, where $W$ is a $K\times K$ Haar random unitary [Eq.~\eqref{eqn:uw}]. The matrix $W$ can be thought of as acting as the source of randomness for the random basis of $\mathcal V^\perp$. The distribution of $W_{NK}W_{NK}^T$ is the setting considered in \cite{shou2026proof}, and so it will be possible to obtain Theorem~\ref{thm:hiding} for $K<M$ by reducing to the hiding case with $K=M$ proved there.

We start with the first part, on relating the distribution of $U_{NK}U_{NK}^T$ to one involving $W_{NK}W_{NK}^T$ for $W$ a $K\times K$ Haar random unitary.
\begin{lem}\label{lem:reduce}
Let $U$ be an $M\times M$ Haar random matrix, and let $U_{NK}$ be its top left $N\times K$ submatrix with $N\le K$. Denote by $V$ the $M\times(M-K)$ matrix consisting of the last $M-K$ columns of $U$, and define the $N\times N$ matrix $P_0:=I_N-E_NVV^\dagger E_N^\dagger$, for $E_N=(I_N\;0_{M-N})$. Then
\begin{align}\label{eqn:uw}
U_{NK}\rvsim P_0^{1/2}W_{NK},\quad\text{ and }\quad U_{NK}U_{NK}^T\rvsim P_0^{1/2}W_{NK}W_{NK}^T(P_0^{1/2})^T,
\end{align}
for $W$ an independent $K\times K$ Haar random unitary and $W_{NK}$ its top $N\times K$ submatrix.
\end{lem}

\begin{proof}
Let $\mathcal V$ be the span of the last $M-K$ columns of $U$. Let $A_0$ be an $M\times K$ matrix whose columns form an orthonormal basis for $\mathcal V^\perp$, and let $W$ be an independent $K\times K$ Haar random unitary. 
Conditioned on $\mathcal V$, we have
\begin{align}
U_{NK}\rvsim E_NA_0W,
\end{align}
for $E_N=(I_N\;0_{M-N})$. Since the distribution of $W$ is invariant under unitary rotations, we would like to pass the $E_N$ through to $W$ (so that we can consider $W_{NK}$), but $A_0$ is rectangle-shaped so we have to do some manipulation.

Note that $A_0A_0^\dagger=I_M-VV^\dagger$ the projection onto $\mathcal V^\perp$,
and so $(E_NA_0)(E_NA_0)^\dagger=I_N-E_NVV^\dagger E_N^\dagger=P_0$. 
We can do polar decomposition (or SVD) to obtain $E_NA_0=P_0^{1/2}R$, for $R$ an $N\times K$ semiunitary matrix;\footnote{In general, $R$ need not be unique. However, in this case one can show $P_0$ is invertible almost surely (a.s.) since $N\le K$, so $R$ is unique a.s. 
To see this, let $A$ be the $M\times K$ matrix consisting of the first $K$ columns of $U$, so $P_0=(E_NA)(E_NA)^\dagger$. For $Z$ an $M\times K$ matrix of iid standard complex Gaussians, we have $E_NA\rvsim E_NZ(Z^\dagger Z)^{-1/2}$ \cite[\S1.2]{Meckes-book}. Note that $Z$ has rank $K$ a.s. (the probability of the next column being in the span of the previous columns is 0), so $\rank Z^\dagger Z=\rank Z=K$ and $Z^\dagger Z$ is invertible a.s. The matrix $E_NZ$ is an $N\times K$ matrix of iid standard complex Gaussians, so has rank $N$ a.s. Thus $\rank[(E_NA)(E_NA)^\dagger]=\rank(E_NA)=N$ a.s., so $P_0$ is invertible a.s.} 
i.e. $RR^\dagger=I_N$ and the rows of $R$ are orthonormal. 
Then we can extend the rows of $R$ to a full orthonormal basis of $\C^K$, which will be given by the rows of a $K\times K$ unitary $\xi$ (independent of $W$), with $R=(I_N\;0_{K-N})\xi$. Since $W$ is invariant under unitary multiplication, and $\xi$ is independent of $W$, conditioned on $\mathcal V$ we get
\begin{align*}
U_{NK}&\rvsim E_NA_0W=P_0^{1/2}(I_N\;0_{K-N})\xi W\\
&\rvsim P_0^{1/2}(I_N\;0_{K-N})W=P_0^{1/2}W_{NK},\numberthis
\end{align*}
where $W_{NK}$ is the top $N\times K$ submatrix of the $K\times K$ unitary $W$.
Thus we obtain \eqref{eqn:uw}.
\end{proof}

\begin{proof}[Proof of Proposition~\ref{prop:hiding-3}]
By Lemma~\ref{lem:reduce}, the distribution of $U_{NK}U_{NK}^T$ is the same as the distribution of $P_0^{1/2}W_{NK}W_{NK}^T(P_0^{1/2})^T$, where $P_0$ is defined as in the lemma and $W$ is an independent $K\times K$ Haar random unitary matrix.
By \cite{shou2026proof}, since $W_{NK}$ is the $N\times K$ submatrix for a $K\times K$ Haar unitary matrix, then for $\Gb\dsim\gsym$ and $N=O(\sqrt{K})$, 
\begin{align}
d_\tv(\sqrt{K}W_{NK}W_{NK}^T,\Gb)\le O(N/\sqrt{K}).
\end{align}
Letting $Q:=(MK^{-1}P_0)^{1/2}$, which is independent of $\Gb$ and $W$, we can then write
\begin{align*}
d_\tv(MK^{-1/2}U_{NK}U_{NK}^T,\Gb)&\le d_\tv(Q\sqrt{K}W_{NK}W_{NK}^TQ^T,Q\Gb Q^T)+ d_\tv(Q\Gb Q^T,\Gb)\\
&\le O(N/\sqrt{K})+d_\tv(Q\Gb Q^T,\Gb).\numberthis\label{eqn:triangle}
\end{align*}
So to prove the theorem it suffices to bound $d_\tv(Q\Gb Q^T,\Gb)$. 
To do this, we will show that $Q$ is typically close to the identity, which will make the TVD small. 

For fixed $Q=\qf$ (e.g. conditioned on $V$), both of the involved distributions $q\Gb q^T$ and $\Gb$ have explicit Gaussian densities, so we can directly estimate the TVD.
For fixed invertible $\qf$, the density of $\qf\Gb \qf^T$ over the space of $N\times N$ symmetric complex matrices is calculated by change of variables\footnote{The map $X\mapsto qXq^T$ is linear, and can be expressed as $(q\otimes q)\operatorname{vec}(X)\equiv\operatorname{vec}(qXq^T)$, where $\operatorname{vec}(X)$ is the vectorization of $X$ formed by stacking columns of $X$. To calculate the Jacobian of the map $X\mapsto qXq^T$, if $q$ is diagonalizable with eigenpairs $\{(\lambda_i,|u_i\rangle)\}_{i=1}^N$, take the eigenbasis $u_iu_j^T+u_ju_i^T=|u_i\rangle\langle \bar u_j|+|u_j\rangle\langle\bar u_i|$ for $i\le j$ of the transformation $X\mapsto qXq^T$, which gives (complex) Jacobian $\prod_{i\le j}\lambda_i\lambda_j=(\det q)^{N+1}$ and (real) Jacobian $|\det q|^{2(N+1)}$.} from $\Gb$ to be 
\begin{align}\label{eqn:fq}
f_\qf(Z)=c_N|\det \qf|^{-2(N+1)}e^{-\frac{1}{2}\Tr \bar \qf^{-1} Z^\dagger (\qf^{-1})^\dagger \qf^{-1}Z(\qf^{-1})^T},
\end{align}
where $c_N=2^{-N}\pi^{-N(N+1)/2}$ is the normalization constant for the density $c_Ne^{-\frac12\Tr Z^\dagger Z}$ of $\Gb$. 

Total variation distance can be bounded using the Kullback--Leibler (KL) divergence, or relative entropy, via Pinsker's inequality, 
\begin{align}\label{eqn:pinsker}
d_\tv(X,Y)&\le\sqrt{\frac12\dkl(X||Y)},
\end{align}
where the KL divergence is $\dkl(X||Y):=\int f(x)\log\frac{f(x)}{g(x)}\,dx$ for $f$ and $g$ the respective density functions for $X$ and $Y$.

For fixed invertible $\qf$, letting $s=\qf^\dagger \qf$, the KL divergence for $\qf\Gb \qf^T$ and $\Gb$ is
\begin{align*}
\dkl(\qf\Gb \qf^T||\Gb)&=\E_{Z\sim q\Gb q^T}[\log|\det \qf|^{-2(N+1)}-\frac12\Tr \bar \qf^{-1} Z^\dagger (\qf^{-1})^\dagger \qf^{-1}Z(\qf^{-1})^T+\frac12 \Tr Z^\dagger Z]\\
&=-(N+1)\log\det s-\frac12[N^2+N]+\frac12[(\Tr s)^2+\Tr(s^2)],\numberthis\label{eqn:dkls}
\end{align*}
using e.g. $\E_{Y\dsim\gsym}[\Tr(Y^\dagger sYs^T)]=(\Tr s)^2+\Tr(s^2)$ by direct expansion of the trace.
Write $s=\qf^\dagger\qf=I_N+\delta$. (We will later show that, for random $Q$,  $\delta$ is small with high probability, so $s$ is typically a small perturbation of the identity.) Expand
\begin{align*}
\dkl(\qf\Gb \qf^T||\Gb)&=-(N+1)\log\det(I_N+\delta)-\frac12[N^2+N]+\frac12[N^2+2N\Tr\delta+(\Tr\delta)^2+N+2\Tr\delta+\Tr(\delta^2)]\\
&=(N+1)[\Tr\delta-\Tr\log(I_N+\delta)]+\frac12(\Tr\delta)^2+\frac12\Tr(\delta^2)\\
&\le (N+3/2)\Tr(\delta^2)+\frac12(\Tr\delta)^2,\quad\text{ if }\|\delta\|\le1/2,\numberthis\label{eqn:dklq}
\end{align*}
using $x-\log(1+x)\le x^2$ for $-1/2\le x$.

Now we return to random $P_0$ and $Q=(MK^{-1}P_0)^{1/2}$.
Let $E_NV=:v$ be the top right $N\times (M-K)$ submatrix of $U$. We will show that $S=Q^\dagger Q=MK^{-1}(I_N-vv^\dagger)$ is typically a small perturbation of the identity. Letting $L:=M-K$ for notational convenience, then
\begin{align*}
S=\frac MKI_N-\frac MKvv^\dagger &=I_N+\frac{L}{K}I_N-\frac MKvv^\dagger.
\end{align*}
Then $\delta=\frac{L}{K}I_N-\frac MKvv^\dagger$, and so $\Tr\delta=\frac{LN}{K}-\frac{M}{K}\Tr(vv^\dagger)$, and 
$\delta^2=\frac{L^2}{K^2}I_N-\frac{2LM}{K^2}vv^\dagger+\frac{M^2}{K^2}vv^\dagger vv^\dagger$.
Taking the expectation over the Haar random unitary $U$, we see that $\E\Tr\delta=0$ since $\E\Tr(vv^\dagger)=\frac{LN}{M}$.
Also, using Weingarten calculus, see e.g. \cite{CollinsSniady2006,Collins2003},
\begin{align*}
\E(\Tr(vv^\dagger))^2=\sum_{i,k=1}^N\sum_{j,\ell=1}^{L}\E|v_{ij}|^2|v_{k\ell}|^2
&=\sum_{i,k=1}^N\sum_{j,\ell=1}^{L}\frac{1}{M^2-1}[1+\delta_{ik}\delta_{j\ell}]-\frac{1}{M(M^2-1)}[\delta_{j\ell}+\delta_{ik}]\\
&=\frac{N^2L^2+NL}{M^2-1}-\frac{N^2L+NL^2}{M(M^2-1)},\numberthis
\end{align*}
and
\begin{align*}
\E\Tr(vv^\dagger vv^\dagger)=\sum_{x,y=1}^N\sum_{i,j=1}^L\E[v_{xi}v_{yj}\bar v_{xj}\bar v_{yi}]
&=\sum_{x,y=1}^N\sum_{i,j=1}^L\frac{1}{M^2-1}[\delta_{ij}+\delta_{xy}]-\frac{1}{M(M^2-1)}[1+\delta_{xy}\delta_{ij}]\\
&=\frac{N^2L+NL^2}{M^2-1}-\frac{N^2L^2+NL}{M(M^2-1)}.\numberthis
\end{align*}

Then
\begin{align*}
\E(\Tr \delta)^2&=\frac{L^2N^2}{K^2}-\frac{2LMN}{K^2}\E\Tr(vv^\dagger)+\frac{M^2}{K^2}\E(\Tr(vv^\dagger))^2\\
&=\frac{LN(M-N)}{K(M^2-1)}
=O\left(\frac{LN}{KM}\right),\numberthis\label{eqn:trdelta2}\\
\E\Tr(\delta^2)&=\frac{L^2N}{K^2}-\frac{2LM}{K^2}\E\Tr(vv^\dagger)+\frac{M^2}{K^2}\E\Tr(vv^\dagger vv^\dagger)\\
&=\frac{LN(MN-1)}{K(M^2-1)}
=O\left(\frac{LN^2}{KM}\right).\numberthis\label{eqn:trdelta22}
\end{align*}
This also implies
\begin{align}\label{eqn:prob12}
\P[\|\delta\|>1/2]&\le 4\E\Tr(\delta^2)=O\left(\frac{LN^2}{KM}\right).
\end{align}

Returning to random $Q$, in order to invoke \eqref{eqn:dklq}, we need to ensure $\|\delta\|\le1/2$, which occurs with probability at least $1-O(N^2/K)$ by \eqref{eqn:prob12}.
For handling the event $\|\delta\|>1/2$, we want to go back to TVD since it is bounded, while KL-divergence need not be.
Let $\delta(q):=q^\dagger q-I_N$.
Define a random variable $\tilde Q$ via
\begin{align*}
\tilde Q=\begin{cases}Q,&\|\delta(Q)\|\le1/2\\
I,&\text{otherwise}
\end{cases};
\end{align*}
then $d_\tv(Q,\tilde Q)\le \P[Q\ne\tilde Q]=\P[\|\delta(Q)\|>1/2]$.
Since $\tilde Q^\dagger\tilde Q=I_N+\delta(\tilde Q)$
and $\delta(\tilde Q)=\delta(Q)\oneb_{\|\delta(Q)\|\le1/2}$, then
\begin{align}\label{eqn:traceqt2}
\begin{aligned}
\E_{\tilde Q}[(\Tr\delta)^2]&\le \E_Q[(\Tr\delta)^2]=O\left(\frac{LN}{KM}\right),\\
\E_{\tilde Q}[\Tr(\delta^2)]&\le \E_Q[\Tr(\delta^2)]=O\left(\frac{LN^2}{KM}\right).
\end{aligned}
\end{align}
We then estimate
\begin{align*}
d_\tv(Q\Gb Q^T,\Gb)&\le d_\tv(Q\Gb Q^T,\tilde Q\Gb\tilde Q^T)+d_\tv(\tilde Q\Gb\tilde Q^T,\Gb)\\
&\le \P[\|\delta(Q)\|>1/2] +\sqrt{\frac12\dkl(\tilde Q\Gb\tilde Q^T||\Gb)}. 
\numberthis\label{eqn:dtvsplit}
\end{align*}
For any $\tilde Q$ which is independent of $\Gb$, the density of $Y=\tilde Q\Gb \tilde Q^T$ can be directly seen to be $\E_{\tilde Q}[f_{\tilde Q}(Z)]$.
Thus using Jensen's inequality with $x\mapsto x\log x$ convex, and that $\dkl(\tilde q\Gb\tilde q^T||\Gb)$ is bounded for example via \eqref{eqn:dklq},
\begin{align*}
\dkl(\tilde Q\Gb \tilde Q^T||\Gb)&= \int \frac{\E_{\tilde{Q}}[f_{\tilde{Q}}(Z)]}{g(Z)}\log\left[\frac{\E_{\tilde{Q}}[f_{\tilde{Q}}(Z)]}{g(Z)}\right]g(Z)\,dZ\\
&\le\int \E_{\tilde{Q}}\left[\frac{f_{\tilde{Q}}(Z)}{g(Z)}\log\left[\frac{f_{\tilde{Q}}(Z)}{g(Z)}\right]\right]g(Z)\,dZ\\
&=\E_{\qf\sim {\tilde{Q}}}[\dkl(\qf\Gb \qf^T||\Gb)].\numberthis\label{eqn:dkljensen}
\end{align*}
Using that $\tilde Q$ is invertible and $\|\delta(\tilde Q)\|\le1/2$ by construction, \eqref{eqn:dklq} and \eqref{eqn:traceqt2} then imply
\begin{align}
\E_{\tilde \qf\sim\tilde Q}[\dkl(\tilde \qf\Gb\tilde \qf^T||\Gb)]
\le O\left(\frac{LN^3}{KM}\right).
\end{align}
Applying this and \eqref{eqn:prob12} in \eqref{eqn:dtvsplit} gives
\begin{align}\label{eqn:tv-final}
d_\tv(Q\Gb Q^T,\Gb)&\le O\left(\sqrt{\frac{(M-K)N^3}{KM}}\right),
\end{align}
which with \eqref{eqn:triangle} implies \eqref{eqn:tvhiding-3}.
\end{proof}

The proof of Theorem~\ref{thm:hiding} for $N=o(\sqrt{K})$ diverges from the above proof by proving a sharper bound on $d_\tv(Q\Gb Q^T,\Gb)$, using $\chi^2$-divergence and a Bakry--\'Emery concentration result \cite{bakry1985diffusions,kolesnikov2016riemannian}. This proof is given in Appendix~\ref{sec:hiding2}. For the rest of the main text of this paper, we will assume the full Theorem~\ref{thm:hiding} is proved.

\section{Proof of Theorem~\ref{thm:ggt}}\label{sec:ggt}

In this section, we prove Theorem~\ref{thm:ggt} on closeness of (properly scaled) $GG^T$ and $\Gb\dsim\gsym$. It suffices to prove this for $N=O(\sqrt{K})$.
The sparse result in \cite[Theorem 1.5]{shou2026proof} shows that for $N\le K$ and $NK=o(M)$, that 
\begin{align}\label{eqn:sparse}
d_\tv(MU_{NK}U_{NK}^T,G_{NK}G_{NK}^T)&\le O\left(\sqrt{\frac{NK}{M}}\right),
\end{align}
where $G_{NK}$ is an $N\times K$ matrix of iid complex standard Gaussians.
We use this result with Theorem~\ref{thm:hiding} (with the full $N=O(\sqrt{K})$ result) to prove Theorem~\ref{thm:ggt}. Choose e.g. $M=K^3$ (any $\omega(K^2)$ will do). Then for $N=O(K^{1/2})$,
\begin{align*}
d_\tv(\Gb,G_{NK}G_{NK}^T/\sqrt{K})&\le d_\tv(\Gb,MK^{-1/2}U_{NK}U_{NK}^T)+d_\tv(MU_{NK}U_{NK}^T,G_{NK}G_{NK}^T)\\
&\le O\left(\frac{N}{\sqrt{K}}\right)+O\left(\sqrt{\frac{NK}{M}}\right)
\le O\left(\frac{N}{\sqrt{K}}\right),
\end{align*}
as desired. \qed

\section{Hardness argument}\label{sec:hardness}

In this section, we prove Theorem~\ref{thm:hardness}. There are two main differences from the boson sampling argument of \cite{aa}: 
\begin{enumerate}
\item We have a choice of squeezing parameter $s$ and a variable photon number $N$.
\item We lack a general instance generating/rejection sampling lemma to hide $X\in\gsym$ exactly as a submatrix of $UI_KU^T$. 
\end{enumerate}
For (1), because the output probabilities \eqref{eqn:hprob} involve the squeezing parameter $s$, the additive error threshold obtained will depend on the choice of $s$. 
We will choose $s$ to minimize this additive error, which will allow us to state the $|\mathrm{GHE}|_\pm^2$ problem in terms of additive error $\varepsilon\cdot\E|\Haf(\Gb)|^2$.
Difference (2) will be resolved by using an approximate instead of exact instance generating method in \fpostbpp. Since $\postbpp\subseteq\bpp^\np$ \cite[\S2]{aa}, this will not change the end complexity class.
Proposition~\ref{prop:inf} shows why one cannot obtain an exact instance generating lemma via rejection sampling as in \cite[Lemma 5.7]{aa} for Gaussian boson sampling with $K=\alpha M$, $0<\alpha<1/2$.
We note that previous GBS hardness arguments avoided the issue in (2) by either working in the sparse squeezer regime \cite{HamiltonGBS2017,KruseGBS2019}, where instance generating follows from instance generating for Fock boson sampling \cite{aa}, or by considering $K=M$ \cite{shou2026proof} in which case the bound to use the rejection sampling of \cite{aa} is easily found to hold, or by considering the variant bipartite Gaussian boson sampling \cite{grier2022complexity,bouland2023complexity}, which uses a different set-up with two-mode squeezed states with parameters tuned to implement a specific matrix.
For the setting here with arbitrary $K$ including $K=\alpha M$, $0<\alpha<1/2$, we will implement an approximate instance generating procedure.

For difference (1), note that for an initial state consisting of $K$ single-mode squeezed states with equal squeezing parameters $s_i=s>0$, and any photon count $N\in2\N$ \cite{KruseGBS2019},
\begin{align}\label{eqn:pn}
\P[N]&=\binom{N/2+K/2-1}{N/2}\frac{\tanh^N s}{\cosh^K s},
\end{align}
where the binomial coefficient is a generalized binomial coefficient allowing for non-integer $K/2$.
The average number of photons for such an input state is $K\sinh^2s$. We will have $N=O(\sqrt{K})$, and will choose squeezing $s$ (to $\poly(N,1/\delta,1/\varepsilon)$ bits) so that, up to exponentially small error, 
\begin{align}\label{eqn:s-choice}
N=K\sinh^2s,\quad\text{which implies}\quad \frac{\cosh^Ks}{\tanh^Ns}=\frac{\left(1+\frac{N}{K}\right)^{K/2}}{\left(\frac{N}{N+K}\right)^{N/2}}=\frac{e^{N/2}K^{N/2}}{N^{N/2}}\exp\left[\frac{N^2}{4K}+o(1)\right],
\end{align}
using that $N^2=O(K)$.
This choice of $s$ can be seen to minimize the ratio $\frac{\cosh^Ks}{\tanh^Ns}$ (e.g. by taking logarithmic derivative), which will appear in the additive error threshold \eqref{eqn:add-error}. 
Also, note that for $N=o(\sqrt{K})$, as $N\to\infty$,
\begin{align}
\P[N]&=\frac{\Gamma(N/2+K/2)}{(N/2)!\Gamma(K/2)}\frac{\left(\frac{N}{N+K}\right)^{N/2}}{\left(1+\frac{N}{K}\right)^{K/2}}=\frac{1+o(1)}{\sqrt{\pi N}}.
\end{align}
Since this is no worse than inverse polynomial, we could postselect on observing exactly $N$ photons, at the cost of repeating the GBS experiment polynomial-many more times.
However, this postselection will actually be unnecessary in the hardness argument, since we will be using Stockmeyer's algorithm \cite{stockmeyer1985approximation} to estimate output probabilities. The main point of choosing $s$ as in \eqref{eqn:s-choice} is to minimize the additive error threshold in \eqref{eqn:add-error}.
We do not use postselection on $N$ or high collision-free probability anywhere in the hardness reduction.

In order to give a hardness reduction, we need to consider finite-precision inputs and outputs to a Turing machine. In particular, the Gaussian and unitary matrices involve real and complex numbers, which must be rounded to finite-precision.
As in \cite[Footnote 19]{aa}, we encode each entry of a unitary matrix in binary to $\poly(M,1/\varepsilon,1/\delta)$ bits. In general, the resulting binary description $v$ is not actually unitary, but we can associate each description $v$ to a unitary $U(v)\in\U(M)$ obtained by Gram--Schmidt orthonormalization of the columns with typically very small error \cite[Lemma 7.2]{aaronson2003algorithms}. We let $\Uf(M)$ denote the set of such finite binary descriptions $v$, and let $U(v)$ denote its associated unitary $U\in\U(M)$. 
Based on the above, we make the following finite-precision implementation assumption, which we expect can be proved via careful accounting with Gram--Schmidt or QR factorization. 
We note that Assumption~\ref{finite}(a) is one possible route, though not the only one, to rigorously implement the finite-precision vs Haar sampling also used in \cite{aa}.

\begin{assumption}\label{finite}
Let $\eta>0$ be a finite precision and $\rho>0$ a failure probability.
There is a bit precision $Q=\poly(M,\log(1/\eta),\log(1/\rho))$, a number of random bits $R=\poly(M,Q)$, and a polynomial-time algorithm which maps $\{0,1\}^R\ni \xi\mapsto v=v_\xi$ for $v$ a finite binary description with associated $M\times M$ unitary $U(v)$ as described above, such that the following hold: 
\begin{enumerate}[(a)]

\item There is a coupling $(v,H)$, with $H\dsim\mathcal H_M$ Haar random, such that for $\xi$ uniform, 
\begin{align*}
\P[\|U(v)-H\|_\mathrm{op}>\eta/4]\le\rho.
\end{align*}

\item From $(v,S)$, where $S$ is any subset $S\subset\{1,\ldots,M\}$, one can compute a finite-precision binary matrix $\hat Z(v,S)$ in deterministic polynomial time such that
\begin{align*}
\|\hat Z(v,S)-(U(v)I_KU(v)^T)_{S,S}\|_\infty\le \eta/2.
\end{align*}
 
\end{enumerate}
As a result, combining (a) and (b), for this coupling and for $S$ an independent uniform random size $N$ subset,
\begin{align}\label{eqn:c-final}
\P[\|\hat Z(v,S)-(HI_KH^T)_{S,S}\|_\infty>\eta]\le \rho.
\end{align}
We let $\mu_{Q,M}$ denote the law of $v_\xi$ for uniform $\xi$.
\end{assumption}
The main property we need from Assumption~\ref{finite} is \eqref{eqn:c-final}, which says we can produce a binary finite-precision matrix $\hat Z(v,S)$ which is generally a good entry-wise approximation to $(HI_KH^T)_{S,S}$. We need the finite-precision matrix so we can apply an \np-oracle in the proof of Lemma~\ref{lem:instance} below.
The failure probability $\rho$ in \eqref{eqn:c-final} accounts for ill-conditioned starting Gaussian matrices which may have poor finite-precision approximation during the coupling procedure. 

Compared to \cite{aa}, which generally referred directly to the continuum unitary and Gaussian distributions after the preliminary finite-precision discussions, we will need to be more precise, and will keep track of the finite-precision rounding.
We will need to use the complexity class \postbpp\ (also called $\bpp_{\sans{path}}$), which is \bpp\ with postselection, defined precisely in e.g. \cite[\S2]{aa}.
The main property we need for this class is $\postbpp\subseteq\bpp^\np$ \cite{han1997threshold,bellare2000uniform}.
We will also use the function/search analogue \fpostbpp.
To resolve the difference (2) from above, we prove (under Assumption~\ref{finite})

\begin{lem}[approximate instance generation]\label{lem:instance}
Consider $X\dsim\gsym$, and let $\delta,\varepsilon$ be error parameters.
Let $K=\Omega(N^2/\delta^2)$, $M=O(\poly(N,1/\delta))$, $N\le K\le M$, and $[M]:=\{1,\ldots,M\}$. Choose a scale $\gamma=2^{-p}$ for some $p=\poly(N,1/\delta,1/\varepsilon)$.
Then there is an $\fbpp^\np$ algorithm $\mathcal A$ (running in $\poly(N,1/\delta,1/\varepsilon)$ time) which, on a sufficiently fine $\poly(N,1/\delta,1/\varepsilon)$-bit precision implementation, takes as input a matrix $X\dsim\gsym$, and outputs either $\bot$ (failure) or a pair $(v,S^*)\in \Uf(M)\times \{\text{size $N$ subsets of $[M]$}\}$. It succeeds with probability $\ge 1-O(\delta)$ over $X$ and $\mathcal A$. Conditioned on succeeding, the output $(v,S^*)$ satisfies 
\begin{enumerate}[(i)]
\item $\|(U(v)I_KU(v)^T)_{S^*,S^*}-M^{-1}K^{1/2}X\|_\infty\le O(\gamma)$, where $\|\cdot\|_\infty$ denotes the $\ell^\infty$ maximum entrywise norm, and $U(v)$ the unitary associated with the finite-precision description $v$.
\item Let $\mathcal L(v,S^*)$ be the law of $(v,S^*)$ averaged over $X$ in successful trials, and let $\mu_{Q,M}\otimes\operatorname{Unif}$ denote the product measure of $\mu_{Q,M}$ on $\Uf(M)$, and the uniform distribution on size $N$ subsets of $[M]$. Then 
\begin{align}\label{eqn:tvd-us}
d_\tv(\mathcal L(v,S^*),\mu_{Q,M}\otimes \operatorname{Unif})\le O(\delta).
\end{align}
\end{enumerate}
\end{lem}

\begin{proof}[Proof of Lemma~\ref{lem:instance}]
Consider the true distributions $X'=M^{-1}K^{1/2}X$ and $H\dsim\mathcal H_M$, and let $Z:=(HI_KH^T)_{S,S}$ for a uniformly random size $N$ subset $S\subset[M]$. Let $R_\gamma$ denote rounding the real and imaginary part of every entry of a matrix to the interval $[j\gamma,(j+1)\gamma)$, $j\in\Z$, containing it.
The overall idea, ignoring finite-precision implementation for now, is: Given $X'$, we want to generate random $(H,S)$ and postselect on the event $R_\gamma(Z)=R_\gamma(X')$, i.e. up to $\gamma$ error, $X'$ appears as the submatrix $Z=(HI_KH^T)_{S,S}$. 
Intuitively, this is a postselection problem so can be done in $\fpostbpp\subseteq\fbpp^{\np}$. However, we need to consider the finite-precision implementation details to properly apply the $\np$ oracle in $\fbpp^{\np}$ \cite{bellare2000uniform}. We do this as follows.

\begin{enumerate}[1.,leftmargin=*]
\item Truncation. First note we can restrict to $X$ with maximum entry size $\|X\|_\infty\le C_1\sqrt{\log(N/\delta)}$ for large enough $C_1$, as this occurs with high probability $1-O(\delta)$ by Gaussian suprema and concentration bounds.\footnote{If $X_t$ is $\sigma^2$-subgaussian for every $t\in T$, then $\P[\sup_{t\in T}X_t\ge\sqrt{2\sigma^2\log|T|}+x]\le e^{-x^2/2\sigma^2}$; see e.g. \cite[\S5]{vanhandel2016apc550}.
In this case, $|T|=N(N+1)/2$.
}
Then let $X'=M^{-1}K^{1/2}X$ conditioned on $\|X\|_\infty\le C_1\sqrt{\log(N/\delta)}$. 
The cutoff prevents arbitrarily large entries, and allows for the required finite-precision implementation. Note that $d_\tv(Z,X')\le d_\tv(Z,M^{-1}K^{1/2}X)+d_\tv(M^{-1}K^{1/2}X,X')\le O(\delta)$, using Theorem~\ref{thm:hiding}, invariance of Haar measure under row permutations, and the TVD coupling inequality. (If $X'\rvsim Y|\event$, then by a coupling argument, $d_\tv(X',Y)\le \P[\event^c]$.) 

\item Finite-precision cells and bounds. Choose finite precision $\eta=2^{-q}$ with say $\eta\le \min(\delta\gamma^2,\delta2^{-N})$, which requires only $q\sim\poly(N,1/\delta,1/\varepsilon)$ bits of precision. Taking $\rho=O(\delta)$ in \eqref{eqn:c-final} of Assumption~\ref{finite}, with probability $1-O(\delta)$ we can produce a finite-precision matrix $\hat Z(v,S)$ such that $\|\hat Z(v,S)-Z\|_\infty\le\eta$. There is also a finite-precision matrix $\hat X'$ such that $\|\hat X'-X'\|_\infty\le\eta$.
We want to show that just like $d_\tv(R_\gamma(Z),R_\gamma(X'))\le d_\tv(Z,X')=O(\delta)$ for the continuum distributions, that
\begin{align}\label{eqn:Rg-goal}
d_\tv(R_\gamma(\hat Z(v,S)),R_\gamma(\hat X'))=O(\delta).
\end{align}
This bound will ensure the language to sample from later in \eqref{eqn:language} is non-empty with high probability, and that the resulting distribution of $(v,S^*)$ satisfies part (ii) of the lemma.

We start by applying the triangle inequality, giving
\begin{multline}\label{eqn:triangle-gamma}
d_\tv(R_\gamma(\hat Z(v,S)),R_\gamma(\hat X'))\\
\le d_\tv(R_\gamma(\hat Z(v,S)),R_\gamma(Z))+d_\tv(R_\gamma(Z),R_\gamma(X'))+d_\tv(R_\gamma(X'),R_\gamma(\hat X')).
\end{multline}
To bound the first and third distances on the right side, we need to bound the probability that $\hat Z(v,S)$ or $\hat X'$ may have moved $\gamma$-cells during the finite-precision implementation/rounding process. Intuitively, because we can take much finer precision $\eta$ than $\gamma$, this probability is very small.
More formally, let $\partial_\eta=\partial_{\eta,\gamma}$ denote the set of matrices with some real or imaginary coordinate within $\eta$ of the grid $\gamma\Z$. Note that for a matrix outside of $\partial_\eta$, moving by $\eta$ cannot move across the coarser $\gamma$-cell boundaries.
The boundary estimate Lemma~\ref{lem:boundary} below then gives $\P[X'\in\partial_\eta]=O(\delta)$, and 
\begin{align*}
\P[Z\in\partial_\eta]&\le \P[X'\in\partial_\eta]+d_\tv(Z,X')=O(\delta).
\end{align*}
If the coupling in Assumption~\ref{finite} is succesful, then $\|\hat Z(v,S)-Z\|_\infty\le\eta$. Thus the rounded cells $R_\gamma(\hat Z(v,S))$ and $R_\gamma(Z)$ can differ only if the coupling in Assumption~\ref{finite} failed, or if $Z\in\partial_\eta$.
Thus by the TVD coupling bound and Assumption~\ref{finite},
\begin{align*}
d_\tv(R_\gamma(\hat Z(v,S)),R_\gamma(Z))&\le \rho+\P[Z\in\partial_\eta]=O(\delta).\numberthis
\end{align*}
Also by TVD coupling lemma,
\begin{align}
d_\tv(R_\gamma(X'),R_\gamma(\hat X'))&\le \P[X'\in\partial_\eta]=O(\delta).
\end{align}
The remaining term in the middle of the right side of \eqref{eqn:triangle-gamma} is bounded using the hiding property Theorem~\ref{thm:hiding} and TVD conditioning property, as $d_\tv(R_\gamma(Z),R_\gamma(X'))\le d_\tv(Z,X')=O(\delta)$. 
This gives \eqref{eqn:Rg-goal}.

\item Uniform $\np$ witness generation with an $\np$ oracle. 
The precise finite-precision problem is now: Given the cell $R_\gamma(\hat X')$, post-select on the event $R_\gamma(\hat Z(v,S))=R_\gamma(\hat X')$. 
Recall we view $\hat Z(v,S)=\hat Z(v_r,S_t)$ as generated via a random-string model $r,t\in\{0,1\}^{\poly(N,1/\delta,1/\varepsilon)}\mapsto (v_r,S_t)$.
Note, in order to sample uniformly from the $\binom{M}{N}$ choices of $S$, which need not divide $2^L$, we actually consider uniform random $t\in\{0,1\}^\ell$ for $\ell=\left\lceil\log_2\binom{M}{N}\right\rceil$, and accept if $t<\binom{M}{N}$, then map each of those $t$ to a subset $S=S_t$.
For the above finite-precision problem, we want to sample a uniform $(r,t)$ from the language
\begin{align}\label{eqn:language}
\mathscr L_{\hat X'}=\{r\in\{0,1\}^{\poly(N,1/\delta,1/\varepsilon)},t\in\{0,1\}^\ell:R_\gamma(\hat Z(v_r,S_t))=R_\gamma(\hat X')\text{ and }t<\binom{M}{N}\},
\end{align}
and then output $(v_r,S_t)$. 
Due to the TVD bound \eqref{eqn:Rg-goal}, and considering the set of $\gamma$-cells $\{c:\text{not attainable by }R_\gamma(\hat Z(v_r,S_t)))\}$, we see $\mathscr L_{\hat X'}$ is nonempty with probability $1-O(\delta)$ over $\hat X'$.
When $\mathscr L_{\hat X'}$ is nonempty, sampling a uniform pair $(r,t)$ can be done in probabilistic polynomial time with an \np\ oracle via \cite{bellare2000uniform}: the language $\mathscr L_{\hat X'}$ is in \np\ (and \p), so if $\mathscr L_{\hat X'}$ is nonempty, then \cite{bellare2000uniform} generates a uniform random witness $(r,t)\in\mathscr L_{\hat X'}$ with success probability at least $0.2$.
Additionally, we can obtain success probability $1-O(\delta)$ using standard amplification with $O(\log\delta^{-1})$ trials in the \fpostbpp\ to $\fbpp^\np$ implementation. 
In total, conditioned on successful output (including $\mathcal L_{\hat X'}$ nonempty, which can be checked by seeing failure or verifying if the output is in $\mathcal L_{\hat X'}$), we obtain a uniform random sample $(v,S^*)$ from the distribution $\mu_{Q,M}\otimes\operatorname{Unif}$ conditioned on the event $R_\gamma(\hat Z(v,S^*))=R_\gamma(\hat X')$, which is generated in $\fbpp^\np$ in time $\poly(N,1/\delta,1/\varepsilon)$ and with success probability $1-O(\delta)$. 

Since $\|\hat Z(v,S^*)-(U(v)I_KU(v)^T)_{S^*,S^*}\|_\infty\le \eta/2$ and $\|X'-\hat X'\|_\infty\le\eta$, we obtain
\begin{align}\label{eqn:ux}
\|(U(v)I_KU(v)^T)_{S^*,S^*}-X'\|_\infty&\le O(\gamma) + O(\eta)=O(\gamma).
\end{align}
Since we can include $\|X\|_\infty>C_1\sqrt{\log(N/\delta)}$ in the failure probability, \eqref{eqn:ux} gives part (i) of the lemma.

\item It remains to prove (ii) of the lemma.
Given a cell $c=R_\gamma(\hat X')$ attainable by $R_\gamma(\hat Z(v_r,S_t))$, successful outputs $(v,S^*)$ are distributed as $\mu_{Q,M}\otimes\Unif$ conditioned on $R_\gamma(\hat Z(w,S))=c$; call this conditional measure $K_c$, for $c$ attainable by $R_\gamma(\hat Z(v,S))$.
The total distribution of $(v,S^*)$ is then $\sum_c \P[R_\gamma(\hat X')=c|\text{success}]K_c$.
We can also decompose $\mu_{Q,M}\otimes\Unif$ as $\sum_c \P[R_\gamma(\hat Z(w,S))=c]K_c$ for $(w,S)\dsim\mu_{Q,M}\otimes\Unif$.

Using the characterizations \eqref{eqn:tvd-f} and \eqref{eqn:tv2} of TVD, we thus obtain
\begin{align*}
d_\tv(\mathcal L(v,S^*),\mu_{Q,M}\otimes\Unif)&= \sup_{f:\|f\|_\infty\le1}\frac12\sum_c(\P[R_\gamma(\hat X')=c|\text{success}]-\P[R_\gamma(\hat Z(w,S))=c])K_c(f)\\
&\le d_\tv(R_\gamma(\hat X')|\text{success},R_\gamma(\hat Z(w,S))) \\
&\le d_\tv(R_\gamma(\hat X'),R_\gamma(\hat Z(w,S)))+O(\delta)=O(\delta),
\end{align*}
using \eqref{eqn:Rg-goal} and that $\P[\text{failure}]=O(\delta)$. This completes the proof of Lemma~\ref{lem:instance}(ii).
\end{enumerate}
\end{proof}

The following lemma was used in the proof of Lemma~\ref{lem:instance}.
\begin{lem}[boundary estimate]\label{lem:boundary}
Let $\eta\le\gamma/2$, and let $\partial_\eta=\partial_{\eta,\gamma}$ denote the set of matrices with some real or imaginary coordinate within $\eta$ of the grid $\gamma\Z$.
Then for $X'=M^{-1}K^{1/2}X$ with $X\dsim\gsym$,
\begin{align}\label{eqn:boundary}
\P[X'\in\partial_\eta]&\le CN^2\eta(\gamma^{-1}+MK^{-1/2}).
\end{align}
If we consider $X'$ conditioned on a probability $1-\epsilon$ event, then the above bound holds with an additional $+\epsilon$ term.
\end{lem}
\begin{proof}
We do a union bound over the $O(N^2)$ independent entries of $X$. For a single real Gaussian $Y\dsim\RN(0,\sigma^2)$ with density function $f_Y$, which is decreasing for $y\ge0$,
\begin{align*}
\P[\operatorname{dist}(Y,\gamma\Z)\le\eta]&=\sum_{j\in\Z}\int_{j\gamma-\eta}^{j\gamma+\eta}f_Y(y)\,dy\\
&\le 2\eta\|f_Y\|_\infty+2\sum_{j=1}^\infty2\eta f_Y(j\gamma-\eta)\\
&\le \frac{2\eta}{\sigma\sqrt{2\pi}}+\frac{4\eta}{\gamma-\eta}\sum_{j=1}^\infty\int_{j\gamma-\gamma}^{j\gamma-\eta} f_Y(y)\,dy\le C\eta(\sigma^{-1}+\gamma^{-1}).\numberthis
\end{align*}
Since entries of $X'$ have standard deviation $\sigma=\Theta(M^{-1}K^{1/2})$, a union bound gives \eqref{eqn:boundary}. The conditioning statement holds since for a random variable $Z$ and event $\event$, letting $Y:\rvsim Z|\event$, then $d_\tv(Z,Y)\le \P[\event^c]$, by using the coupling bound. (Set $Y=Z$ on $\event$, and $Y\dsim\mathcal L(Z|\event)$ on $\event^c$.)
\end{proof}

Because we round to $\poly(N,1/\delta,1/\varepsilon)$ bits, we want to check the resulting hafnian is not changed too much.
Suppose $\|X\|_\infty,\|Y\|_\infty\le B$. Then for $N\in2\N$,
\begin{align*}
\left||\Haf X|^2-|\Haf Y|^2\right|&=\Bigg|\sum_{\pi,\pi'\in\mathcal P_2(N)}\left(\prod_{\{i,j\}\in\pi}X_{ij}\prod_{\{i',j'\}\in\pi'}\bar X_{i'j'}-\prod_{\{i,j\}\in\pi}Y_{ij}\prod_{\{i',j'\}\in\pi'}\bar Y_{i'j'}\right)\Bigg|\\
&\le \sum_{\pi,\pi'\in\mathcal P_2(N)}NB^{N-1}\|X-Y\|_\infty\\
&\le (N-1)!!^2NB^{N-1}\|X-Y\|_\infty.\numberthis\label{eqn:hafl}
\end{align*}
The rounding process changes a matrix by at most $\|X-Y\|_\infty\le\gamma=2^{-\poly(N,1/\delta,1/\varepsilon)}$. Since $B=O(M^{-1}K^{1/2}\sqrt{\log(N/\delta)})$ for the sampled $U$ in Lemma~\ref{lem:instance}, $\gamma$ is easily chosen so that \eqref{eqn:hafl} is $\le 2^{-\poly(N,1/\delta,1/\varepsilon)}$ for such matrices.

\begin{proof}[Proof of Theorem~\ref{thm:hardness}]

Suppose we had an oracle $\mathcal O$ for limited instances of approximate Gaussian boson sampling as described in the hypotheses. $\mathcal O$ takes as input a string $r\in\{0,1\}^{\poly(M,1/\delta,1/\varepsilon)}$, a  finite binary description $v$ representing an $M\times M$ unitary matrix $U=U(v)$, a squeezing parameter $s>0$ for the first $K=K(M)$ input modes, and an error bound $\varepsilon>0$. We may write the Gaussian boson sampler description as $A=A(v,K,s)$. Over uniform random $r$ (representing randomness in the Gaussian boson sampling output), $\mathcal O$ outputs a distribution $\mathcal D_{\mathcal O}(A,\varepsilon)$ which is $\varepsilon$-close to $\mathcal D_A$, the exact Gaussian boson sampling output distribution for linear optical unitary $U(v)$.

Let $X\dsim\gsym$ be an input matrix and let $\varepsilon,\delta>0$ be the given error parameters. 
Note that it suffices to solve $|\mathrm{GHE}|_\pm^2$ with probability at least $1-O(\delta)$, since we can simply rescale $\delta$ to get the probability to $1-\delta$.
So we want to approximate $|\Haf(X)|^2$ to within additive error $\varepsilon\cdot (N-1)!!$ with success probability at least $1-O(\delta)$ over $X$. 

Choose $(K(M),M)$ with $M=\poly(K(M))$, and squeezing parameter $s$, so that
\begin{align}\label{eqn:nk}
N\le c_1\delta K^{1/2},\qquad N=K\sinh^2s,
\end{align}
for a constant $c_1$ chosen so that the TVD error bound in \eqref{eqn:tvd-us} is e.g. $\le\delta/8$. 
Rescaling $X':=M^{-1}K^{1/2}X$, we then want to approximate $|\Haf(X')|^2$ to within additive error $\varepsilon\cdot(N-1)!!(M^{-1}K^{1/2})^{N}$.

By Lemma~\ref{lem:instance}, for a sufficiently fine numerical precision, with probability $1-O(\delta)$, we can efficiently generate a random finite binary description $v=v_\xi$ of an $M\times M$ unitary $U=U(v)$, and a size $N$ subset $S^*$ of $[M]$, such that 
\begin{align}
\|(UI_KU^T)_{S^*,S^*}-X'\|_\infty\le2^{-\poly(N,1/\delta,1/\varepsilon)},\quad\text{and}\quad d_\tv(\mathcal L(v,S^*), \mu_{Q,M}\otimes\operatorname{Unif})\le O(\delta).
\end{align}
Recall $\mu_{Q,M}$ is the measure on the finite binary descriptions $\Uf(M)$ defined in Assumption~\ref{finite}. 

Conditioned on success of the above instance generating procedure, take $v$ and send $A=A(v,K,s)$ to the GBS oracle $\mathcal O$. 
We will also allow an adversary to know $N$, since it will not affect the argument, and since they could possibly already guess a narrow range of $N$ we are interested in based on the squeezing parameters provided.
Let $\beta$ be an error bound which is polynomial in $\varepsilon$ and $\delta$. Then more formally, denoting $0^{1/\beta}=0\cdots0$ with $\lceil1/\beta\rceil$ zeros, we send the input $\langle A(v,K,s),0^{1/\beta},r\rangle$, for $r\in\{0,1\}^{\poly(M,1/\delta,1/\varepsilon)}$ a random string, to the oracle $\mathcal O$.
As $r$ is varied, $\mathcal O$ returns a sample from $\mathcal D_A'$ with $\|\mathcal D_A-\mathcal D_A'\|_\tv\le\beta$. 

For collision-free $S$ with $N$ photon counts, let 
\begin{align*}
p_S(v):=\P_{\mathcal D_A}[S], \quad q_S(v):=\P_{\mathcal D_A'}[S]=\P_{r}[\mathcal O(A(v,K,s),0^{1/\beta},r)=S].
\end{align*}
The probability $q_{S^*}$ will be approximated in $\fbpp^{\np^{\mathcal O}}$ by Stockmeyer's algorithm \cite{stockmeyer1985approximation} as usual.
The probability $p_{S^*}$ is
\begin{align}
p_{S^*}&=\frac{\tanh^N(s)}{\cosh^K(s)}|\Haf[(UI_KU^T)_{S^*,S^*}]|^2.
\end{align}
Given $p_{S^*}$ or a good approximation to $p_{S^*}$, this with \eqref{eqn:hafl} and Lemma~\ref{lem:instance}(i) will then let us estimate $|\Haf(X')|^2$.

We want to show that $p_{S^*}$ and $q_{S^*}$ are close with high probability over $X$ and $v$.
We can average over the $\binom{M}{N}$ possible $N$-photon collision-free outputs $S$ as in \cite[\S5.2]{aa} to obtain
\begin{align}
\E_{|S|=N}[|p_S-q_S|]&\le\frac{\sum_{|S|=N}|p_S-q_S|}{\binom{M}{N}}
\le \frac{2\|\mathcal D_A-\mathcal D'_A\|_\tv}{\binom{M}{N}}<C_1\beta\frac{N!}{M^N},
\end{align}
using that $M\ge K\ge c_1^{-2}\delta^{-2}N^2$. Then taking $\beta=\varepsilon\delta$, Markov's inequality gives
\begin{align}
\P_{|S|=N}\left[|p_S-q_S|>\frac{\varepsilon}{2}\frac{N!}{M^N}\right]&\le\frac{O(\beta)}{\varepsilon} \le O(\delta).
\end{align}
We will use Lemma~\ref{lem:instance}(ii) to obtain a bound on $|p_{S^*}-q_{S^*}|$.
First, if $(w,S)\dsim \mu_{Q,M}\otimes\operatorname{Unif}$, then since $S$ is uniform, the above implies
\begin{align}
\P_{(w,S)}\left[|p_S(w)-q_S(w)|>\frac{\varepsilon}{2}\frac{N!}{M^N}\right]&\le O(\delta).
\end{align}
Applying Lemma~\ref{lem:instance}(ii), we have 
\begin{align*}
\P_{(v,S^*)}\left[|p_{S^*}(v)-q_{S^*}(v)|>\frac{\varepsilon}{2}\frac{N!}{M^N}\right]&\le \P_{(w,S)}\left[|p_S(w)-q_S(w)|>\frac{\varepsilon}{2}\frac{N!}{M^N}\right]+O(\delta)\\
&\le O(\delta).\numberthis \label{eqn:deltaS}
\end{align*}

Next, as in \cite[\S5.2]{aa}, we use Stockmeyer's algorithm to approximate $q_{S^*}$. 
For any $\alpha>0$, Stockmeyer's algorithm \cite{stockmeyer1985approximation} applied with $f(r):=\mathbf{1}[\mathcal O(A,0^{1/\beta},r)=S^*]$, and standard output probability amplification, gives an estimate $\tilde q_{S^*}$ in $\poly(M,1/\alpha)$ time such that
\begin{align}\label{eqn:qs-stockmeyer}
\P[|\tilde q_{S^*}-q_{S^*}|>\alpha q_{S^*}]<\frac{1}{2^M}.
\end{align}
Also, by Markov's inequality with $\E_{|S|=N}[q_S]\le\binom{M}{N}^{-1}\le\frac{2N!}{M^N}$, and using the TVD bound in Lemma~\ref{lem:instance}(ii) like in \eqref{eqn:deltaS}, we have for any $j>1$,
\begin{align}\label{eqn:qs2}
\P_{(v,S^*)}\left[q_{S^*}>2j\cdot\frac{N!}{M^N}\right]&\le\frac{1}{j}+O(\delta).
\end{align}
Thus taking $j=4/\delta$ and $\alpha=\varepsilon\delta/16$, and also using \eqref{eqn:deltaS}, we get
\begin{align*}
\P\left[|\tilde q_{S^*}-p_{S^*}|>\varepsilon\frac{N!}{M^N}\right]&\le \P\left[|\tilde q_{S^*}-q_{S^*}|>\frac{\varepsilon}{2}\frac{N!}{M^N}\right]+\P\left[| q_{S^*}-p_{S^*}|>\frac{\varepsilon}{2}\frac{N!}{M^N}\right]\\
&\le \P\left[|\tilde q_{S^*}-q_{S^*}|>\alpha q_{S^*}\right]+\P\left[q_{S^*}>2j\frac{N!}{M^N}\right]+O(\delta)\\
&\le \frac{1}{2^M}+\frac{\delta}{4}+O(\delta).\numberthis
\end{align*}
Since $M\ge K\ge\Omega(N^2/\delta^2)$, we get $2^{-M}=O(\delta)$.
We combine this with the $O(\delta)$ failure probability from Lemma~\ref{lem:instance}. 
In total, using the approximation $\tilde q_{S^*}$, with probability at least $1-O(\delta)$ we can approximate $|\Haf[(UI_KU^T)_{S^*,S^*}]|^2$, and thus also $|\Haf(X')|^2$, to additive error 
\begin{align}\label{eqn:add-error}
\varepsilon\cdot \frac{\cosh^K(s)}{\tanh^N(s)}\frac{N!}{M^N}+O(2^{-\poly(N,1/\delta,1/\varepsilon)})&\le \varepsilon\frac{K^{N/2}}{M^N}\frac{N^{N/2}}{e^{N/2}}\sqrt{2\pi N}(1+O(\delta)+O(1/N)),
\end{align}
recalling we chose $s$ in \eqref{eqn:s-choice} up to exponentially small error, and have a small rounding error from \eqref{eqn:hafl}. 
Comparing to
\begin{align*}
\varepsilon\cdot\E|\Haf(X')|^2=\varepsilon \frac{(N-1)!!K^{N/2}}{M^{N}}&= \varepsilon\frac{K^{N/2}}{M^N}\frac{N^{N/2}}{e^{N/2}}\sqrt{2}(1+o(1)),
\end{align*}
we see \eqref{eqn:add-error} differs by a factor of $\sqrt{\pi N}$.
But $\varepsilon$ can absorb $\poly(N)$ factors (just run the procedure with $\varepsilon'=\varepsilon/(c\sqrt{N})$ which adds only a $\poly(\sqrt{N})$ factor), so we obtain the theorem.
\end{proof}

\section{Proof of Proposition~\ref{prop:inf}}\label{sec:inf}

In this section, we prove Proposition~\ref{prop:inf}, showing that the density ratio for $MK^{-1/2}U_{NK}U_{NK}^T$ and $\Gb$ can be exponentially large when $K=\alpha M$, $0<\alpha<1/2$.
We first consider $N=1$, and prove
\begin{lem}\label{lem:f1density}
Let $K\ge2$ and $M>K$. The probability density function for $MK^{-1/2}U_{1K}U_{1K}^T$ over $\C$ is
\begin{align}\label{eqn:f1density}
f_1(z)=\frac{K(K-1)}{2\pi M^2 B(K,M-K)} \int_{M^{-1}K^{1/2}|z|}^1(1-x)^{M-K-1}(x^2-M^{-2}K|z|^2)^{(K-3)/2}\,dx,
\end{align}
where $B(x,y)=\frac{\Gamma(x)\Gamma(y)}{\Gamma(x+y)}$ denotes the beta function.
\end{lem}
\begin{proof}
The scalar quantity $U_{1K}U_{1K}^T$ is distributed as $u_1^2+\cdots+u_K^2$ for a random complex unit vector $u=(u_1,\ldots,u_M)\in\Sc^{M-1}$. For a vector $v$, we will use notation $v_{a:b}$ to indicate the vector $(v_a,v_{a+1},\ldots,v_b)$. Let $u=w/\|w\|_2$ for $w\dsim\CN(0,I_M)$ a standard complex Gaussian vector, so that
\begin{align*}
u_1^2+\cdots+u_K^2&=\frac{w_1^2+\cdots+w_K^2}{\|w_{1:K}\|_2^2}\frac{\|w_{1:K}\|_2^2}{\|w\|_2^2}=:SR,
\end{align*}
where $S=\frac{w_1^2+\cdots+w_K^2}{\|w_{1:K}\|_2^2}\dsim\sum_{j=1}^Ks_j^2$ for $s=(s_1,\ldots,s_K)\dsim\Unif(\Sc^{K-1})$, and $R:=\frac{\|w_{1:K}\|_2^2}{\|w\|_2^2}\dsim\Beta(K,M-K)$. Moreover, $S$ and $R$ are independent, since $(w_1,\ldots,w_K)/\|w_{1:K}\|_2$ is independent of $\|w_{1:K}\|_2$, and also of $w_{K+1:M}$. 
The density functions $f_R$ for $R\dsim\Beta(K,M-K)$ on $[0,1]$, and
$f_S$ for $S$ on the unit disk, are
\begin{align}
f_R(x)&=\frac{1}{B(K,M-K)}x^{K-1}(1-x)^{M-K-1},\;\text{ and } \;f_S(z)=\frac{K-1}{2\pi}(1-|z|^2)^{(K-3)/2},
\end{align}
where $B(x,y)=\frac{\Gamma(x)\Gamma(y)}{\Gamma(x+y)}$ is the beta function.
The density function $f_S$ was derived in \cite[Eq.~(3.12)]{pereyra1983marginal}; alternatively, it can be derived by writing $S=\sum_{j=1}^Ks_j^2$, for $s=x+iy$ with $x,y\dsim\RN(0,I_K)$ independent real Gaussian vectors, and then writing $|S|$ in terms of the eigenvalues of the real Wishart matrix $\begin{pmatrix}\|x\|_2^2&x\cdot y\\x\cdot y&\|y\|_2^2\end{pmatrix}=VV^T$ for $V=\binom{-x-}{-y-}$, since the density function for Wishart eigenvalues is well-known.

Since $S$ and $R$ are independent, the density function $f_{SR}$ for $SR$ on the unit disk is, e.g. by doing change of variables for $\E[h(SR)]$ for arbitrary $h$,
\begin{align*}
f_{SR}(z)&=\int_{|z|}^1dx\,f_R(x)f_S(z/x)\frac{1}{x^2}\\
&=\frac{K-1}{2\pi B(K,M-K)}\int_{|z|}^1dx\,(1-x)^{M-K-1}(x^2-|z|^2)^{(K-3)/2}.\numberthis
\end{align*}
With scaling factors, the density function $f_1$ of $MK^{-1/2}U_{1K}U_{1K}^T$ is then given by \eqref{eqn:f1density}.
\end{proof}

We can then give
\begin{proof}[Proof of Proposition~\ref{prop:inf}]
We start with the case $N=1$.
The density function of $1\times1$ $\Gb$ is $g_1(z)=\frac{1}{2\pi} e^{-|z|^2/2}$, and the density of $MK^{-1/2}U_{1K}U_{1K}^T$ is given by \eqref{eqn:f1density} of Lemma~\ref{lem:f1density}.

For Laplace estimation, we consider $\frac{f_1(z)}{g_1(z)}=\frac{K(K-1)}{M^2 B(K,M-K)}\int_{M^{-1}K^{1/2}|z|}^1 e^{F_{z}(x)}\,dx$, with
\begin{align}\label{eqn:Fz}
F_{z}(x)&=(M-K-1)\log(1-x)+\frac{K-3}{2}\log\left(x^2-M^{-2}K|z|^2\right)_++\frac{|z|^2}{2}.
\end{align}
Solving $F_z'(x)\equiv\frac{(K-3)x}{x^2-M^{-2}K|z|^2}-\frac{M-K-1}{1-x}=0$ and taking the positive root gives critical point
\begin{align}\label{eqn:xzroot}
x_z&=\frac{K-3+\sqrt{(K-3)^2+4(M-4)(M-K-1)M^{-2}K|z|^2}}{2(M-4)}.
\end{align}
If $x_z$ is bounded away from the limits of integration as $M\to\infty$, then one can check that Laplace's method\footnote{Since we only need a lower bound, we can actually skip Laplace's method (to avoid having to check precise conditions for error bounds), and just integrate over an $O(M^{-1/2})$ neighborhood of $x_*$, in which $F_{z_*}(x_*)\ge F_{z_*}(x)-O(1)$ by Taylor's theorem.} gives
\begin{align*}
\frac{f_1(z)}{g_1(z)}&=\frac{K(K-1)}{M^2 B(K,M-K)}\int_{M^{-1}K^{1/2}|z|}^1e^{F_z(x)}\,dx\\
&=\frac{K(K-1)}{M^2 B(K,M-K)}\sqrt{\frac{2\pi}{|F_z''(x_z)|}}e^{F_z(x_z)}\left(1+O(1/M)\right),\numberthis\label{eqn:fglaplace}
\end{align*}
with an implicit constant which may depend on $\alpha$.

We can choose any $z$ with $|z|\le MK^{-1/2}$ to try to make a large density ratio, so let's solve for critical points of $F_z(x_z)$ over $|z|$. Writing $\phi(x,z):=F_z(x)$, the multivariate chain rule gives $\frac{d}{d|z|}\phi(x_z,z)=\partial_x\phi(x_z,z)\frac{dx_z}{d|z|}+\partial_z\phi(x_z,z)$. Since $\partial_x\phi(x_z,z)=F_z'(x_z)=0$, we solve
\begin{align*}
(\partial_zF_z)(x_z)&\equiv\frac{K-3}{2}\frac{-2M^{-2}K|z|}{(x_z^2-M^{-2}K|z|^2)}+|z|=0,
\end{align*}
which has nonzero solutions
\begin{align}\label{eqn:zstar}
|z_*|^2:=M^2K^{-1}x_{z_*}^2-(K-3).
\end{align}
Let $x_*:=x_{z_*}$, so $F'_{z_*}(x_*)=0$.
Plugging \eqref{eqn:zstar} into the equation $F'_{z_*}(x_*)=0$ implies
\begin{align*}
(K-3)x_*(1-x_*)&=(M-K-1)M^{-2}K(K-3),
\end{align*}
so
$x_*(1-x_*)=(M-K-1)M^{-2}K$,
and $x_*=\frac{1\pm\sqrt{1-4(M-K-1)M^{-2}K}}{2}$.
We will take the larger root and recall $\alpha<1/2$, which gives
\begin{align*}
x_*&=1-\alpha+\frac{\alpha}{(1-2\alpha)M}+O(1/M^2),\numberthis\label{eqn:xstar}\\
|z_*|^2&=\frac{M}{\alpha}(1-\alpha)^2-M\alpha+\frac{2(1-\alpha)}{1-2\alpha}+3+O(1/M).\numberthis\label{eqn:zstar2}
\end{align*}
We also see that $M^{-1}K^{1/2}|z_*|=\sqrt{1-2\alpha}+O(1/M)$, so for $\alpha>0$ there is a constant gap separating $x_*=1-\alpha+O(1/M)$ from the lower limit of integration in \eqref{eqn:fglaplace} as $M\to\infty$.
For this $z_*$ and $x_*$,
\begin{align*}
F_{z_*}(x_*)&=(M-K-1)\log(1-x_*)+\frac{K-3}{2}\log(M^{-2}K(K-3))+\frac12|z_*|^2\\
&=M\left[\log\alpha+\frac{1-2\alpha}{2\alpha}\right]-4\log\alpha+O(1/M).\numberthis\label{eqn:fstar}
\end{align*}
For the second derivative term in \eqref{eqn:fglaplace}, evaluating the (negative of the) second derivative of $F_{z_*}$ at $x_*$ using \eqref{eqn:xstar} and \eqref{eqn:zstar2} gives
\begin{align*}
-F_{z_*}''(x_*)&=\frac{M-K-1}{(x_*-1)^2}+\frac{(K-3)(x_*^2+M^{-2}K|z_*|^2)}{(x_*^2-M^{-2}K|z_*|^2)^2}\\
&=M\frac{2-3\alpha}{\alpha^3}+O(1).\numberthis\label{eqn:fderiv2}
\end{align*}
The remaining term to expand in \eqref{eqn:fglaplace} is the beta function $B(K,M-K)=\frac{\Gamma(K)\Gamma(M-K)}{\Gamma(M)}$.
Using Stirling's formula $\log\Gamma(y)\sim(y-\frac12)\log y-y+\frac12\log(2\pi)+O(1/y)$, with $K,M-K\to\infty$, we see 
\begin{align*}
\log B&(K,M-K)\\
&=K\log\frac{K}{M}+(M-K)\log\frac{M-K}{M}+\frac12\log(2\pi)+\frac12\log\frac{M}{K(M-K)}+O(1/M)\\
&=M\left[\alpha\log\alpha+(1-\alpha)\log(1-\alpha)\right]-\frac12\log M-\frac12\log(\alpha(1-\alpha))+\frac12\log(2\pi)+O(1/M).\numberthis\label{eqn:betaexp}
\end{align*}
In total, plugging \eqref{eqn:fstar}, \eqref{eqn:fderiv2}, and \eqref{eqn:betaexp} into \eqref{eqn:fglaplace}, we obtain 
\begin{align}
\log\frac{f_1(z_*)}{g_1(z_*)}&=M\left[(1-\alpha)\log\frac{\alpha}{1-\alpha}+\frac{1-2\alpha}{2\alpha}\right]+\frac12\log(1-\alpha)-\frac12\log(2-3\alpha)+O(1/M).
\end{align}
Thus
\begin{align*}
\sup_z\frac{f_1(z)}{g_1(z)}&\ge \sqrt{\frac{1-\alpha}{2-3\alpha}}\exp\left({M\left[(1-\alpha)\log\frac{\alpha}{1-\alpha}+\frac{1-2\alpha}{2\alpha}\right]}\right)(1-O(1/M)).\numberthis\label{eqn:ratiosup}
\end{align*}
For $0<\alpha<1/2$, we can check that
\begin{align}
(1-\alpha)\log\frac{\alpha}{1-\alpha}+\frac{1-2\alpha}{2\alpha}>0,
\end{align}
for example by setting $y=\frac{1-\alpha}{\alpha}$ and noting the above is equivalent to $\frac{y^2-1}{2y}-\log y>0$ for $y>1$, i.e. $\alpha<1/2$, and that this holds by checking the derivative is $>0$ for $y>1$.
Then \eqref{eqn:ratiosup} gives
\begin{align}
\sup_z\frac{f_1(z)}{g_1(z)}&\ge \Omega_\alpha(e^{c_\alpha M}),
\end{align}
with the implicit constant and $c_\alpha$ depending on $\alpha=K/M$.
Moreover, for $\alpha\in[\eta,1/2-\eta]$, the constants can be taken uniform depending only on $\eta$.

The $N=1$ case implies the higher dimensional cases since $f_1$ and $g_1$ are marginal densities of the higher dimensional densities $f$ and $g$. 
If $\|f/g\|_\infty<\infty$, then letting $d\hat Z$ denote integration over all variables except the top left entry $Z_{11}$, we see
\begin{align}
f_1(Z_{11})=\int d\hat Z\,\frac{f(Z)}{g(Z)}{g(Z)}&\le \|f/g\|_\infty\int d\hat Z\,g(Z)=\|f/g\|_\infty g_1(Z_{11}).
\end{align}
Thus $\|f_1/g_1\|_\infty\le \|f/g\|_\infty$.
\end{proof}

\appendix

\section{Proof of Theorem~\ref{thm:hiding} with \texorpdfstring{$N=o(\sqrt{K})$}{N=o(sqrt(K))}}\label{sec:hiding2}

In Section~\ref{sec:hiding}, we proved Proposition~\ref{prop:hiding-3}, which is a weaker version of Theorem~\ref{thm:hiding}. In this section, we prove the full Theorem~\ref{thm:hiding}. The general proof idea is still the same as for Proposition~\ref{prop:hiding-3}, in particular using the key Lemma~\ref{lem:reduce}. The difference is only in the later part of the proof, where we will do a more technical estimate to bound $d_\tv(Q\Gb Q^T,\Gb)$.

Proposition~\ref{prop:hiding-3} works for any $N=o(K^{1/3})$, and the final bound \eqref{eqn:tv-final} shows the proof also works for $N=o(K^{1/2})$ if $L=M-K=O(M^{1/2})$. In this section, we improve the allowed size to $N=o(K^{1/2})$ for any $L$. 
Since we already have the desired result $N=o(K^{1/2})$ when $L=O(M^{1/2})$, and since $N=O(K^{1/2})=O(M^{1/2})$, we will always consider $L\ge N$ in this section. It suffices to prove the required bound in Theorem~\ref{thm:hiding} for $N^2\le c_0K$ with $c_0$ chosen sufficiently small. For convenience, we will also take $K\ge 2N$.

The proof starts in the same way as in Section~\ref{sec:hiding}, but the difference is we will bound the quantity $d_\tv(Q\Gb Q^T,\Gb)$ more sharply. 
Specifically, in this section we prove the stronger bound (compare to \eqref{eqn:tv-final})
\begin{align}\label{eqn:dtv-error}
d_\tv(Q\Gb Q^T,\Gb)&\le O\left(\frac{N}{\sqrt{K}}\right),
\end{align} 
using an exact $\chi^2$-divergence expression, and a standard Bakry--\'Emery concentration result \cite{bakry1985diffusions,kolesnikov2016riemannian}. Avoiding the inefficiency from Jensen's inequality with the KL-divergence in \eqref{eqn:dkljensen} will allow us to obtain the error bound \eqref{eqn:dtv-error}, at the cost of a more involved proof. Putting this new bound into \eqref{eqn:triangle} will then imply Theorem~\ref{thm:hiding}.

It will be useful to view symmetric $N\times N$ Gaussian matrices $\Gb$ and $q\Gb q^T$ directly as multivariate Gaussians in $d:=N(N+1)/2$ variables.
To this end, let $\sym_N:=\{X\in\C^{N\times N}:X^T=X\}\cong\C^{N(N+1)/2}$, with the inner product $\langle X|Y\rangle=\frac12\Tr(X^\dagger Y)$.
The usual orthonormal basis for $\sym_N$ consists of $|i\rangle\langle j|+|j\rangle\langle i|$ for $1\le i<j\le N$, and $\sqrt{2}|i\rangle\langle i|$ for $i=1,\ldots,N$.
For an invertible covariance operator $\Sigma:\sym_N\to\sym_N$, the circular-symmetric complex Gaussian with covariance $\Sigma$ has density on $\sym_N$, with respect to the above orthonormal basis\footnote{Note, due to the choice of inner product and resulting basis vector normalization, this differs by a constant factor from \eqref{eqn:fq}, which gives the density with respect to the upper triangular coordinates of the matrix.}, given by 
\begin{align}\label{eqn:density-sym}
f(Z)&= |\det\!{}_{\sym_N} \Sigma|^{-1}\pi^{-d}e^{-\frac12\Tr(Z^\dagger\Sigma^{-1}(Z))}.
\end{align}
We see, due to the orthonormal basis, that $\Sigma=I$ corresponds to the distribution of $\Gb\dsim\gsym$, whose entries have variance 2 on the diagonal and 1 off of the diagonal.

We can define a family of covariance operators $\Sigma_R$ on $\sym_N$ via $\Sigma_R(X):=RXR^T$, for $N\times N$ hermitian matrices $R$.
From a computation like \eqref{eqn:fq}, we see that for $q\ge0$, $q\Gb q^T$ has covariance operator $\Sigma_{q^2}$. Also, when $R>0$, then $\Sigma_R>0$ as well, since if $(r_i,\varphi_i)_i$ are eigenpairs of $R$, then $(r_ir_j, |\varphi_i\rangle\langle\bar\varphi_j|+|\varphi_j\rangle\langle\bar\varphi_i|)_{i\le j}$ are eigenpairs of $\Sigma_R$. 

As in Section~\ref{sec:hiding}, let $P_0=I_N-E_NVV^\dagger E_N^\dagger$, where $E_N=(I_N\;0_{M-N})$ and $V$ is the $M\times (M-K)$ matrix consisting of the last $M-K$ columns of the Haar random $U$. Set $Q=(MK^{-1}P_0)^{1/2}$ and $S=Q^\dagger Q=Q^2$.
For fixed $Q=q$, $q\Gb q^T$ viewed as a Gaussian on $\sym_N$ then has covariance $\Sigma_{s}$, for $s=q^\dagger q=q^2$.
Write $s=q^2=:I+\delta$ as in Section~\ref{sec:hiding}, and let $\Sigma_s=I_{\sym_N}+A$ with $A=A(\delta)=\Sigma_s-I_{\sym_N}$.
As before, it will be useful to ensure $\|\delta\|$ is bounded away from 1. For random $Q=(MK^{-1}P_0)^{1/2}$ and $\delta=Q^2-I$, define the event
\begin{align*}
\event:=\{\|\delta\|\le1/4\},\quad\text{which has }\P[\event^c]=O\left(\frac{LN^2}{KM}\right),
\end{align*}
by the same argument as \eqref{eqn:prob12}.
We will bound the distance between $\mu_\event=\mathcal L(Q\Gb Q^T|\event)$ and $\gamma=\mathcal L(\Gb)$.
Since $\P[\event^c]=O(\frac{LN^2}{KM})$, this will also give us an adequate bound on the total variation distance between $\mu=\mathcal L(Q\Gb Q^T)$ and $\mathcal L(\Gb)$.

It will be useful to use $\chi^2$-divergence.
The $\chi^2$-divergence between probability measures $\mu$ and $\nu$ is $\chi^2(\mu||\nu):=\int\big(\frac{d\mu}{d\nu}-1\big)^2\,d\nu=\int\big(\frac{d\mu}{d\nu}\big)^2\,d\nu-1$, and so $d_\tv(\mu,\nu)=\frac12\int|\frac{d\mu}{d\nu}-1|\,d\nu\le\frac12\sqrt{\chi^2(\mu||\nu)}$. The point is that for Gaussian mixture $\mu_\event$ and standard Gaussian $\gamma$ on $\sym_N$, we can evaluate an exact formula for the $\chi^2$-divergence.

The density of $\mu_\event=\mathcal L(Q\Gb Q^T|\event)$ is given as follows. 
Let $f_A$ be the density as in \eqref{eqn:density-sym} for the complex Gaussian on $\sym_N$ with covariance matrix $\Sigma=I_{\sym_N}+A$.
Letting $A$ be distributed as the random variable $A(\delta)=\Sigma_{I+\delta}-I_{\sym_N}$ conditioned on $\event$, the density for $\mu_\event$ is then given by $\E_A f_{A}$ (for example, one can use Fubini--Tonelli).
Then
\begin{align*}
1+\chi^2(\mu_\event||\gamma)=\int_{\sym_N}\left(\frac{d\mu_\event}{d\gamma}\right)^2\,d\gamma 
&=\int_{\sym_N}\frac{(\E_Af_{A}(Z))^2}{f_0^2(Z)}\,f_0(Z)\,dZ\\
&=\int_{\sym_N}\frac{\E_A f_A(Z)\E_B f_B(Z)}{f_0(Z)}\,dZ,\numberthis
\end{align*}
where $B$ is an independent copy of $A$. Since $f_A,f_B,f_0$ are Gaussian densities, we can evaluate, using \eqref{eqn:density-sym},
\begin{align*}
1+\chi^2(\mu_\event||\gamma)&=\E_{A,B}\int_{\sym_N}\pi^{-d}\det(I+A)^{-1}\det(I+B)^{-1}e^{-\frac12\Tr(Z^\dagger[(I+A)^{-1}+(I+B)^{-1}-I](Z))}\,dZ\\
&=\E_{A,B}\frac{\det(I+A)^{-1}\det(I+B)^{-1}}{\det((I+A)^{-1}+(I+B)^{-1}-I)}
=\E_{A,B}\frac{1}{\det(I-AB)},\numberthis\label{eqn:chi-det}
\end{align*}
the last equality by pulling out factors $\det(I+A)^{-1}$ and $\det(I+B)^{-1}$ on the left and right sides of the denominator, and provided all inverses are defined and $(I+A)^{-1}+(I+B)^{-1}-I>0$, which we now verify.

To check $(I+A)^{-1}$ and $(I+B)^{-1}$ exist, we can estimate $\|A\|$ on $\event$. Recall $S=I+\delta$ which is self-adjoint, and let $(\lambda_j,\varphi_j)_{j=1}^N$ be the eigenpairs of $\delta$.
Then one can check the (not necessarily normalized) eigenpairs of $\Sigma_S$ are $((1+\lambda_i)(1+\lambda_j),|\varphi_i\rangle\langle\bar\varphi_j|+|\varphi_j\rangle\langle\bar\varphi_i|)_{i\le j}$. On $\event=\{\|\delta\|\le1/4\}$, the absolute value of the eigenvalues of $A=\Sigma_S-I$ are thus
\begin{align}\label{eqn:Anorm}
|(1+\lambda_i)(1+\lambda_j)-1|&=|\lambda_i+\lambda_j+\lambda_i\lambda_j|\le \frac{9}{16}\quad\Longrightarrow\quad \|A\|\le\frac{9}{16}<1.
\end{align}
Similarly, we can use the above bound to see $(I+A)^{-1}+(I+B)^{-1}-I\ge (\frac{16}{25}+\frac{16}{25}-1)I=\frac{7}{25}I>0$.
Thus we do have in total
\begin{align}\label{eqn:chi2-det}
1+\chi^2(\mu_\event||\gamma)&=\E_{A,B}\det(I-AB)^{-1}.
\end{align}

The rest of the proof will be to bound $\E_{A,B}\det(I-AB)^{-1}$. We will do this one expectation at a time. Letting $A=A(\delta)$ and taking a logarithm, define
\begin{align}
F_B(\delta):=-\log\det(I-A(\delta)B),
\end{align}
so that $\E_A\det(I-AB)^{-1}=\E_A e^{F_B(\delta)}$.
Note that $\det(I-AB)$ is real and positive on $\event$, using the factorization in \eqref{eqn:chi-det} into ratios of determinants of positive operators on $\event$.

We will use the Bakry--\'Emery criterion \cite{bakry1985diffusions} to bound $\E_A e^{F_B(\delta)}$ via a concentration inequality.
We use the version from \cite[Theorem 2.1]{kolesnikov2016riemannian}, which handles sets with not-necessarily-smooth boundaries.
\begin{thm}[Bakry--\'Emery criterion \cite{bakry1985diffusions,kolesnikov2016riemannian}]\label{thm:bakry-emery}
Suppose $\Omega$ is a convex subset of $\R^m$ (whose interior is thus geodesically convex, i.e. there is a distance-minimizing geodesic between any two points).
Let $\mu$ be a probability measure of the form $d\mu(x)=e^{-V(x)}\,dx$ for $V$ smooth on the interior of $\Omega$. 
Recall the entropy functional for nonnegative $h$ is $\operatorname{Ent}_\mu(h):=\int h\log h\,d\mu-\int h\,d\mu\cdot \log\int h\,d\mu$.
If $D^2V_x[H,H]\ge\lambda\|H\|_\hs^2$ for all $x\in\Omega$, then for all $f\in C^1(\Omega)$,
\begin{align}
\operatorname{Ent}_\mu(f^2)&\le\frac{2}{\lambda}\int|\nabla f|^2\,d\mu.
\end{align}
\end{thm}

By the standard Herbst argument, the entropy bound with $f=e^{tF/2}$ implies for $L$-Lipschitz $F$ (see e.g. \cite[\S3.1.2]{Wainwright}),
\begin{align}\label{eqn:conc}
\log \E_\mu e^{t(F-\E F)}&\le \frac{t^2L^2}{2\lambda},\quad t\in\R.
\end{align}

We apply this to $F=F_B$ and $V$ the negative log-density of $\delta$ restricted to $\Omega$ which will be a convex subset of $\event$. 
We bound the Hessian of $V$ and the Lipschitz constant for $F_B$ as follows.
For $V$, recall $\delta=S-I=MK^{-1}P_0-I$, where $P_0=(E_NA_1)(E_NA_1)^\dagger$ for $E_NA_1$ the top left $N\times K$ submatrix of the Haar random unitary $U$.
Then for $N\le K,L$, $P_0$ follows a complex matrix-variate beta distribution with density $f(P)\propto(\det P)^{K-N}\det(I-P)^{L-N}\oneb_{0<P<1}$; see the argument for the real orthogonal case in \cite[Lemma 2.12]{Meckes-book}, which in the complex case gives the complex matrix-variate beta density in \cite[\S2]{diaz-garcia2010complex}. 
This gives the negative log-density of $\delta$ as 
\begin{align}
V(\delta)=-(K-N)\log\det(I+\delta)-(L-N)\log\det(L/K-\delta)+C',
\end{align}
for some constant $C'$ depending on $K,L,M,N$, and with the domain $\oneb_{0<P<I}$ replaced by $\oneb_{-I<\delta<L/K}$. Thus we take $\Omega:=\event\cap\{-I<\delta<L/K\}$, which is convex.

The necessary matrix derivatives can be calculated using the formulas \cite{cookbook}
\begin{align*}
\partial_t \det A(t)&=\det A(t) \Tr(A(t)^{-1}\partial_tA(t)),\\
\partial_t \Tr A(t)&= \Tr\partial_tA(t),\\ 
\partial_t A^{-1}&=-A^{-1}(\partial_t A)A^{-1}. % (59)
\end{align*}
This gives
\begin{align*}
DV(\delta)[H]&=\partial_{t=0}V(\delta+tH)=-(K-N)\Tr((I+\delta)^{-1}H)+(L-N)\Tr((L/K-\delta)^{-1}H).
\end{align*}
Taking another derivative using $\delta\mapsto\delta+tH$, we get for $H$ hermitian,
\begin{align*}
D^2V(\delta)[H,H]&=(K-N)\Tr((I+\delta)^{-1}H(I+\delta)^{-1}H)+(L-N)\Tr((L/K-\delta)^{-1}H(L/K-\delta)^{-1}H) \\
&\ge (K-N)\Tr((I+\delta)^{-1}H(I+\delta)^{-1}H)\\
&\ge (K-N)\left(\frac{4}{5}\right)^2\|H\|_\hs^2\ge cK\|H\|_\hs^2,\quad\text{on $\event$,}\numberthis\label{eqn:V-hessian}
\end{align*}
and since we assume $K\ge2N$.
For the Lipschitz constant for $F_B(\delta)=-\log\det(I_{\sym_N}-A(\delta)B)$, we similarly compute
\begin{align}\label{eqn:FB-hessian}
DF_B(\delta)[H]&=\partial_{t=0}F_B(\delta+tH)=\Tr[(I-AB)^{-1}[DA(\delta)[H]]B].
\end{align}
Since we can also compute $DA(\delta)[H](X)=HX(1+\delta)^T+(1+\delta)XH^T$,
we can check that as an operator on $\sym_N$, 
\begin{align}\label{eqn:DA}
\|DA(\delta)[H]\|_\hs\le C\sqrt{N}\|H\|_\hs, 
\end{align}
for example by calculating the Hilbert--Schmidt norms in the usual basis for $\sym_N$.
From \eqref{eqn:FB-hessian}, and $\|(I-AB)^{-1}\|\le c$ and \eqref{eqn:Anorm}, we then get
\begin{align}
\begin{aligned}
\|\nabla F_B\|&=\sup_{H=H^\dagger:\|H\|_\hs=1}|DF_B(\delta)[H]| \le C\sqrt{N}\|B\|_\hs. 
\end{aligned}
\end{align}
Combining with \eqref{eqn:V-hessian} and taking the convex set $\Omega=\event\cap\{-I<\delta<L/K\}$, Theorem~\ref{thm:bakry-emery} and \eqref{eqn:conc} with $t=1$ thus give
\begin{align}\label{eqn:bl1}
\E_A e^{F_B}&\le\exp\left[\E_A F_B+\frac{CN}{K}\|B\|_\hs^2\right].
\end{align}
To bound $\E_A F_B$, we use the power series expansion $F_B(\delta)=-\log\det(I-A(\delta)B)=\sum_{j=1}^\infty\frac1j\Tr((AB)^j)$. For $j\ge2$,
\begin{align}
|\Tr((AB)^j)|&\le \|AB\|^{j-2}\|AB\|_\hs^2 \le \left(\frac{9}{16}\right)^{2(j-2)}\Tr(A^2B^2). 
\end{align}
Thus $F_B(\delta)\le \Tr(AB)+C\Tr(A^2B^2)$, and it will next be enough to bound $\E_A A$ and $\E_A A^2$, which is done by the following lemma.

\begin{lem}[expectation values]\label{lem:id}
Let $A=A(\delta)=\Sigma_{I+\delta}-I_{\sym_N}$, conditioned on the event $\event$. Then there are scalars $a,b$ such that
\begin{align}
\E A&=aI,\quad \E A^2=bI,\quad
\text{with }a=O\bigg(\frac{\sqrt{N}}{K}\bigg), \quad b=O\left(\frac{N}{K}\right).
\end{align}
\end{lem}
This lemma will be proved in Section~\ref{subsec:proof-lem-id}. Applying it to \eqref{eqn:bl1} and the discussion below it, we obtain
\begin{align}
\E_A\det(I-AB)^{-1}&\le e^{q(B)},\quad\text{ for }q(B)=a\Tr B+\left(b+\frac{CN}{K}\right)\Tr(B^2).
\end{align}

Now we apply Theorem~\ref{thm:bakry-emery} with \eqref{eqn:conc} again, this time to $\E_B e^{q(B)}$.
We calculate the gradient of $q(B)$, letting $DB[H]$ denote $DB(\delta)[H]$, which satisfies the bound \eqref{eqn:DA},
\begin{align}
Dq(B(\delta))[H]&=a\Tr(DB[H])+2(b+CN/K)\Tr[B(\delta)DB[H]]. 
\end{align}
Using \eqref{eqn:DA}, $b=O(\frac{N}{K})$, $\|B\|_\hs\le \sqrt{d}\|B\|=O(N)$ from \eqref{eqn:Anorm}, and $|\Tr_{\sym_N} Y|\le cN\|Y\|_\hs$ and $|\Tr XY|\le\|X\|_\hs\|Y\|_\hs$, we see on $\event$,
\begin{align}
\|\nabla q(B)\|&\le O\left(|a|N^{3/2}+\frac{N^{5/2}}{K}\right). 
\end{align}
Then Theorem~\ref{thm:bakry-emery} with \eqref{eqn:conc}, and $a=O(\frac{\sqrt{N}}{K})$, give
\begin{align}
\E_B e^{q(B)}&\le \exp\left[\E_B q(B)+\frac{C}{K}\left(\frac{N^4}{K^2}+\frac{N^5}{K^2}\right)\right].
\end{align}
Using Lemma~\ref{lem:id} for $B$, along with the bounds on $a=O(\frac{\sqrt{N}}{K})$ and $b=O(\frac{N}{K})$, then gives 
\begin{align}
\E_Bq(B)&\le da^2+d\frac{C'N}{K}b\le O\left(\frac{N^3}{K^2}+\frac{N^4}{K^2}\right)=O\left(\frac{N^4}{K^2}\right),
\end{align}
since $N^2\le c_0K$ by assumption. Thus
\begin{align}
\E_{A,B}\det(I-AB)^{-1}&\le \E_Be^{q(B)}\le \exp\left(O\bigg(\frac{N^4}{K^2}\bigg)\right).
\end{align}

Finally, returning to \eqref{eqn:chi2-det}, we get for $N^2\le c_0K$,
\begin{align}
\chi^2(\mu_\event||\gamma)&\le e^{O\left(\frac{N^4}{K^2}\right)}-1 \le O\left(\frac{N^4}{K^2}\right).
\end{align}
Since $d_\tv(\mu_\event,\gamma)\le\frac12\sqrt{\chi^2(\mu_\event||\gamma)}$, and $d_\tv(\mu_\event,\mu)\le \P[\event^c]\le O(\frac{LN^2}{KM})$ for $\mu=\mathcal L(Q\Gb Q^T)$ (for example, use a coupling $Z=Q\Gb Q^T\dsim\mu$, $Y=Z$ on $\event(Q)$, $Y\dsim\mu_\event$ on $\event(Q)^c$), we obtain 
\begin{align}
d_\tv(Q\Gb Q^T,\Gb)&\le O\left(\frac{N}{\sqrt{K}}\right).
\end{align}
Inserting this in \eqref{eqn:triangle} gives Theorem~\ref{thm:hiding}. \qed

\subsection{Proof of Lemma~\ref{lem:id}}\label{subsec:proof-lem-id}
In this section, we prove Lemma~\ref{lem:id}. It will be an application of the following result, followed by some trace bounds similar to those in Section~\ref{sec:hiding}.
\begin{lem}\label{lem:iso}
Let $R$ be a random hermitian matrix with law invariant under $R\mapsto VRV^\dagger$ for any $V\in\U(N)$. Then $\E\Sigma_R=a_RI$ for some $a_R\in\R$.
\end{lem}
\begin{proof}[Proof of Lemma~\ref{lem:iso}]
It will be enough to show that $\E\Sigma_R(|x\rangle\langle x|)=a|x\rangle\langle x|$ for every $x\in\R^N$, since the real basis elements $|e_i\rangle\langle e_j|+|e_j\rangle\langle e_i|$ for $\sym_N$ can be expressed as $|e_i+e_j\rangle\langle e_i+e_j|-|e_i\rangle\langle e_i|-|e_j\rangle\langle e_j|$. 
We start with $|x\rangle\langle x|=|e_1\rangle\langle e_1|=\begin{pmatrix}1&\mathbf 0\\\mathbf 0&\mathbf 0_{N-1}\end{pmatrix}$. We see that $\E\Sigma_R(|e_1\rangle\langle e_1|)$ is invariant under the unitary congruence $(\cdot)\mapsto D(\cdot)D^T$ by the block matrix $D:=\operatorname{diag}(1,W)$ for any $W\in\U(N-1)$, since noting that $D|e_1\rangle\langle e_1|D^T=|e_1\rangle\langle e_1|$, then
\begin{align*}
DR|e_1\rangle\langle e_1|R^TD^T&=DRD^\dagger D|e_1\rangle\langle e_1|D^T(D^T)^\dagger R^TD^T\\
&\rvsim R|e_1\rangle\langle e_1|R^T.\numberthis\label{eqn:wphase}
\end{align*}
Write $\E\Sigma_R(|e_1\rangle\langle e_1|)\equiv\E R|e_1\rangle\langle e_1|R^T=\begin{pmatrix}a&b\\b^T&C\end{pmatrix}$ for $a\in\R$ and $C$ an $(N-1)\times(N-1)$ complex symmetric matrix.
Then taking $W=e^{it}I_{N-1}$ for $t\in\R$ in \eqref{eqn:wphase} shows $b=0$ and $C=0$, so $\E\Sigma_R(|e_1\rangle\langle e_1|)=a|e_1\rangle\langle e_1|$. For general unit $x\in\R^N$, we can then simply let $V\in\mathrm{O}(N)$ be an orthogonal matrix such that $V|e_1\rangle=|x\rangle$; then $V^\dagger=V^T$ and
\begin{align}
\Sigma_R(|x\rangle\langle x|)=R|x\rangle\langle x|R^T&=RV|e_1\rangle\langle e_1|V^T R^T\rvsim VR|e_1\rangle\langle e_1|R^TV^T,
\end{align}
and so $\E \Sigma_R(|x\rangle\langle x|)=Va|e_1\rangle\langle e_1|V^T=a|x\rangle\langle x|$, as desired.
\end{proof}

We apply this to estimate $\E_\event A$ and $\E_\event A^2$ in Lemma~\ref{lem:id}.
Since $S=1+\delta=MK^{-1}P_0=MK^{-1}(I_N-E_NVV^\dagger E_N^\dagger)$ for $V$ the last $M-K$ columns of an $M\times M$ Haar random matrix, one can check that $S\rvsim \xi_NS\xi_N^\dagger$ for any $\xi_N\in\U(N)$. (For example, one can use $\xi_NE_N=[\xi_N\;0_{M-N}]=E_N(\xi_N\oplus I_{M-N})$, following by invariance of $V$ under multiplication by $(\xi_N\oplus I_{M-N})\in\U(M)$.) Moreover, since conjugation preserves eigenvalues and $S=I+\delta$, the distributions conditioned on $\event$ are also the same.
From the lemma, then $\E_\event A=\E_\event\Sigma_S-I=aI$ for some $a$, and $\E_\event A^2=\E_\event\Sigma_{S^2}-2\E_\event\Sigma_S+I=bI$ for some $b$. 
\begin{itemize}[leftmargin=*]

\item For $\E_\event A^2=bI$, in terms of the eigenvalues $\lambda_j$ for $\delta$ (see \eqref{eqn:Anorm}), we have $\Tr_{\sym_N}(A^2)=\sum_{i\le j}(\lambda_i+\lambda_j+\lambda_i\lambda_j)^2\le \sum_{i\le j}C(\lambda_i^2+\lambda_j^2)$. Thus
\begin{align*}
b=\frac{1}{\dim\sym_N}\E_\event\Tr_{\sym_N}(A^2) 
&\le \frac{2}{N(N+1)}\E_\event\sum_{i\le j}C(\lambda_i^2+\lambda_j^2)\\
&\le \frac{C'}{N+1}\E_\event\Tr\delta^2=O\left(\frac{LN}{KM}\right),\numberthis\label{eqn:b}
\end{align*}
where we used \eqref{eqn:trdelta22} and that $\P[\event]\ge 1-c$ to conclude for any $X\ge0$, $\E[X|\event]=\E[X\oneb_{\event}]/\P[\event]\le C\E[X]$.

\item For $\E_\event A=aI$,  we have $a=\frac{2}{N(N+1)}\E_\event\Tr_{\sym_N}A$. Letting $(s_i)_i$ be the eigenvalues of $S$, then the eigenvalues of $\Sigma_S$ are $(s_is_j)_{i\le j}$ (see the discussion before \eqref{eqn:Anorm}), and so
\begin{align}
\Tr_{\sym_N}\Sigma_S&=\sum_{i\le j}s_is_j=\frac12[(\Tr S)^2+\Tr(S^2)].
\end{align}
Since $A=\Sigma_S-I$ and $S=I+\delta$, we then have
\begin{align}\label{eqn:TrA}
\Tr_{\sym_N}A=\frac12[(\Tr\delta)^2+\Tr(\delta^2)]+(N+1)\Tr\delta.
\end{align}
Using $\E\Tr\delta=0$, we can bound $\E_\event\Tr\delta$ as 
\begin{align*}
|\E_\event\Tr\delta|&=\frac{|\E\Tr\delta\oneb_{\event^c}|}{\P[\event]}\le C\sqrt{\E(\Tr\delta)^2}\sqrt{\P[\event^c]}.
\end{align*}
Dividing \eqref{eqn:TrA} by $\dim\sym_N=\frac{N(N+1)}{2}$ and using \eqref{eqn:trdelta2}, \eqref{eqn:trdelta22} to bound the unconditioned expectation values of $(\Tr\delta)^2$ and $\Tr(\delta^2)$ terms, and \eqref{eqn:prob12} for $\P[\event^c]\le O(\frac{LN^2}{KM})$, we get
\begin{align}
|a|\le O\bigg(\frac{L\sqrt{N}}{KM}\bigg)=O\bigg(\frac{\sqrt{N}}{K}\bigg).
\end{align}
\end{itemize}
This proves Lemma~\ref{lem:id}. \qed

\vspace{2mm}
\noindent
\textbf{Acknowledgments.}
This project used GPT-5.5 Thinking and Pro and GPT-5.6 Sol for coming up with proof ideas and methods, as well as for general checking and proofreading. The paper was written by the authors and all results and proofs were checked and validated by the authors, who are fully responsible for the final content. We thank Joseph Iosue and Yu-Xin Wang for useful discussions.  L.S.\ and A.V.G.~acknowledge support from the U.S.~Department of Energy, Office of Science, Accelerated Research in Quantum Computing, Fundamental Algorithmic Research toward Quantum Utility (FAR-Qu). L.S.\ and A.V.G.~were also supported in part by ARL (W911NF-24-2-0107), ONR MURI, and NSF QLCI (award No.~OMA-2120757). V.G. was supported by the US Army Research Office under Grant Number W911NF-23-1-0241.

\bibliographystyle{amsalpha_edit}
\bibliography{hiding.bib}

\end{document}